\documentclass[11pt]{article}

\usepackage{amsfonts}
\usepackage{amsmath}  
\usepackage{graphicx, comment, color, hyperref, framed, amsthm}
\usepackage{enumitem}
\usepackage{hyperref, comment}
\usepackage{mathrsfs}

\usepackage[a4paper]{geometry}
\def\Cov{\text{\rm Cov}}

\usepackage{tikz}

\makeatletter
\newtheorem*{rep@theorem}{\rep@title}
\newcommand{\newreptheorem}[2]{%
\newenvironment{rep#1}[1]{%
 \def\rep@title{#2 \ref{##1}}%
 \begin{rep@theorem}}%
 {\end{rep@theorem}}}
\makeatother

\newtheorem{theorem}{Theorem}%[section]
\newtheorem{lemma}[theorem]{Lemma}

\newtheorem{proposition}[theorem]{Proposition}

\newtheorem{definition}[theorem]{Definition}
\newtheorem{example}[theorem]{Example}

\newreptheorem{theorem}{Theorem}

\newcommand{\DSBS}{\text{\rm DSBS}}
\newcommand{\fR}{\mathfrak{R}}
\newcommand{\fS}{\mathfrak{S}}

\newcommand{\dd}{\text{\rm{d}}}
\newcommand{\Var}{\text{\rm{Var}}}

\newcommand{\E}{\mathbb{E}}

\definecolor{dred}{rgb}{0.7,0,0}
\definecolor{darkblue}{rgb}{0, .07, .5}
\definecolor{darkred}{rgb}{0.5,0,0}
 \definecolor{mahogany}{rgb}{0.65, 0., 0.5}

\begin{document}

\title{$\Phi$-Entropic Measures of Correlation}
\author{Salman Beigi and Amin Gohari\\\small{
School of Mathematics, Institute for Research in Fundamental Sciences (IPM), Tehran, Iran}\\\small{
Department of Electrical Engineering, Sharif University of Technology, Tehran, Iran}
}
\date{}
\maketitle

\begin{abstract}

A measure of correlation is said to have the \emph{tensorization} property if it is unchanged when computed for i.i.d.\ copies. More precisely, a measure of correlation between two random variables $(X, Y)$ denoted by $\rho(X, Y)$, has the tensorization property if $\rho(X^n, Y^n)=\rho(X, Y)$ where $(X^n, Y^n)$ is $n$ i.i.d.\ copies of $(X, Y)$.
%Measures of correlation that satisfy the tensorization property have applications in information theory, notably in the distribution simulation problem as well as in channel coding with correlation sources. 
Two well-known examples of such measures are the maximal correlation and the hypercontractivity ribbon (HC~ribbon).  We show that the maximal correlation and HC ribbons are special cases of $\Phi$-ribbon, defined in this paper for any function $\Phi$ from a class of convex functions ($\Phi$-ribbon reduces to HC~ribbon and the maximal correlation for special choices of $\Phi$).  Any $\Phi$-ribbon is shown to be a  measures of correlation with the tensorization property.  We show that the $\Phi$-ribbon also characterizes the $\Phi$-strong data processing inequality constant introduced by Raginsky.  
We further study the $\Phi$-ribbon for the choice of $\Phi(t)=t^2$ and introduce an equivalent characterization of this ribbon. 

\end{abstract}

% to do: application section
% strong subadditivity
%conclusion.

%*********************************************************

\tableofcontents

\section{Introduction}

A measure of correlation is called to have the \emph{tensorization} property if it is unchanged when computed for i.i.d.\ copies. 
Such measures of correlations have found applications in the non-interactive distribution simulation problem~\cite{KamathAnantharam}, distributed source and channel coding problems~\cite{KangUlukus}, as well as simulation of non-local correlation by wirings~\cite{OurPaper}. In this paper we introduce new measures of correlation with the tensorization property that generalize two previously known such measures.

Let us explain the notion of  tensorization via the example of non-interactive distribution simulation~\cite{KamathAnantharam}. Fix some bipartite distribution $p_{XY}$. Suppose that two parties, Alice and Bob, are given i.i.d.\ samples $X^n$ and $Y^n$ respectively, and they are asked to output \emph{one} sample of $A$ and $B$ respectively, distributed according to some predetermined distribution $q_{AB}$. Alice and Bob can choose $n$ to be as large as they want, but are not allowed to communicate after receiving $X^n$ and $Y^n$. The problem of deciding whether this task is feasible or not is a hard problem in general. Nevertheless, we may obtain impossibility results using the data processing inequality.  

Suppose that $I(X^n; Y^n) < I(A; B)$. In this case, by the data processing inequality, local transformation of $(X^n, Y^n)$ to $(A, B)$ is infeasible. However, note that mutual information is \emph{additive}, i.e., we have $I(X^n; Y^n) = n I(X; Y)$. Then, unless $X$ and $Y$ are independent,  by choosing $n$ to be large enough, $I(X^n; Y^n)$ becomes as large as we want and greater than $I(A; B)$. Therefore, the data processing inequality of mutual information does not give us any useful bound on this problem, simply because mutual information is additive and increases when computed on i.i.d.\ copies. So we need to use a measure of correlation with the tensorization property as defined below. 
 
Suppose that there is some function $\rho(\cdot, \cdot)$ of bipartite distributions that similar to mutual information satisfies the data processing inequality (i.e., it is a measure of correlation), but instead satisfies
\begin{align}\label{eq:def-tensor-0}
\rho(X^n, Y^n) = \rho(X, Y).
\end{align}
The above equation is called the \emph{tensorization} property. Given such a measure  we find that local transformation of $(X^n, Y^n)$ to $(A, B)$ is impossible (even for arbitrarily large $n$) if $\rho(X, Y)< \rho(A, B)$.

\paragraph{Maximal correlation:} A notable example of such a measure of correlation  is \emph{maximal correlation}~\cite{Hirschfeld, Gebelein, Renyi1, Renyi2}, which was used by Witsenhausen~\cite{Witsenhausen75} in his extension of the result of G\'acs and K\"orner on common information~\cite{GacsKorner}.
Maximal correlation $\rho(X, Y)$ of  a bipartite probability distribution $p_{XY}$ is the maximum of Pearson's correlation coefficient over all non-constant functions $f$ and $g$ of $X$ and $Y$ respectively. That is,
\begin{align}\label{eq:max-correlatoin-31}
\rho(X, Y)=& \max\frac{\E\big[(f(X)-\E[f(X)])(g(Y)-\E[g(Y)])\big]}{\sqrt{\Var[f(X)]\Var[g(Y)]}},
\end{align}
where $\E[\cdot]$ and $\Var[\cdot]$ are expected value and variance respectively; moreover, the maximum is taken over all non-constant functions $f=f(X)$ and $g=g(Y)$. 
%Maximal correlation can equivalently be written as
%\begin{align*}
%\rho(X_1, X_2) =  \max  ~& \E[f_{X_1}\,g_{X_2}]\\
%&\E_{X_1} [f] = \E_{X_2} [g]=0, \\
%& \E [f^2]= \E [g^2] =1.
%\end{align*}
%
%We always have $0\leq \rho(X_1, X_2)\leq 1$. Moreover, $\rho(X_1, X_2)=0$ if and only if $X_1$ and $X_2$ are independent, and $\rho(X_1, X_2)=1$ if and only if $X_1$ and $X_2$ have an explicit \emph{common data}. 
Maximal correlation satsifies the following two important properties: 
\begin{enumerate}[label=(\roman*)]
\item \emph{Tensorization:} We have
\begin{align}\rho(X_1X_2, Y_1Y_2) = \max\{\rho(X_1, Y_1), \rho(X_2, Y_2)\},\label{eqn:amnewtensor}\end{align}
when $X_1Y_1$ and $X_2Y_2$ are independent, i.e., $p_{X_1X_2Y_1Y_2}=p_{X_1Y_1}\cdot p_{X_2Y_2}$. This equation in particular gives~\eqref{eq:def-tensor-0}.
\item \emph{Monotonicity:} We have
\begin{align}\rho(A, B) \leq \rho(X, Y),\label{eqn:dataprocessingrho}\end{align}
when $A- X- Y- B$ forms a Markov chain. Thus, maximal correlation can be thought of as a \emph{measure of correlation}
\end{enumerate}

\paragraph{Maximal correlation ribbon:}
Another measure of correlation that satisfies the tensorization property is the \emph{maximal correlation ribbon} (MC~ribbon) defined in~\cite{OurPaper}.  MC~ribbon $\fS(X, Y)$ is the set of $(\lambda_1, \lambda_2)\in[0,1]^2$ such that
\begin{align}\label{eq:mc-ribbon-def-i}
\Var[f]\geq \lambda_1\Var_{X}\big[\mathbb{E}[f|X]\big]+\lambda_2\Var_{Y}\big[\mathbb{E}[f|Y]\big],
\end{align}
for all functions $f=f(X, Y)$ of both $X$ and $Y$. It is shown in~\cite{OurPaper} that the MC ribbon satisfies the following properties:
\begin{enumerate}[label=(\roman*)]
\item \emph{Tensorization:} $\fS(X_1X_2, Y_1Y_2) =\fS(X_1, Y_1)\cap \fS(X_2, Y_2) $
when $X_1Y_1$ and $X_2Y_2$ are independent.
\item \emph{Monotonicity:} $\fS(X, Y)\subseteq \fS(A, B)$ 
when $A- X- Y- B$ forms a Markov chain. 
\end{enumerate}
Thus the MC~ribbon satisfies properties similar to those of maximal correlation. Indeed it is shown in~\cite{OurPaper} that the maximal correlation can be characterized in terms of the MC~ribbon: 
\begin{align}\label{eq:rho2-mc-ribbon}
\rho^2(X, Y)=\inf \frac{1-\lambda_1}{\lambda_2},
\end{align} 
over all $(\lambda_1, \lambda_2)\in \fS(X, Y)$ with $\lambda_2\neq 0$. Thus the MC~ribbon is a parent invariant of bipartite correlations which also characterizes $\rho(X, Y)$. Moreover, as will be done later in this paper, the definition~\eqref{eq:mc-ribbon-def-i} can easily be generalized to the multivariate case (more than two random variables).

\paragraph{$\Phi$-entropy:} 
Variance of a function is equal to its \emph{$\Phi$-entropy} when we take $\Phi(t)=t^2$. To explain this, note that for a function $\Phi$, the $\Phi$-entropy of $f=f(X)$ is defined by
$$H_\Phi(f):= \E[\Phi(f)] - \Phi(\E f).$$
Then for $\Phi(t)=t^2$ we have $H_\Phi(f)=\Var[f]$. Moreover, the MC~ribbon is equal to the set of $(\lambda_1, \lambda_1)\in [0,1]^2$ such that for all functions $f=f(X, Y)$ we have
$$H_\Phi(f) \geq \lambda_1 H_\Phi\big(\E[f|X]\big) + \lambda_2 H_\Phi\big(\E[f|Y]\big).$$ 
This expression for MC~ribbon suggests generalizing it for arbitrary choices of $\Phi$, or at least for convex ones. This idea would seem more reasonable once we note that another important measure of correlation that satisfies the tensorization properties, namely the \emph{hypercontractivity ribbon} (HC~ribbon), can also be expressed in the above form for the choice of $\Phi(t) = 1-h((1+t)/2)$ where $h(\cdot)$ is the binary entropy function: $h(p) = -p\log p - (1-p)\log(1-p)$ (this fact is explained in details later in Example \ref{ex:1}). As a result, the two most well-known measures of correlation that satisfy the tensorization property can be expressed in terms of $\Phi$-entropy as above.

\paragraph{Our contributions:} 
Following the above ideas, for any convex function $\Phi$ we define a $\Phi$-ribbon associated to any $k$ (correlated) random variables $(X_1, \dots, X_k)$. We prove that $\Phi$-ribbon satisfies the tensorization property as well as the monotonicity property similar to the MC~ribbon assuming that $\Phi$ satisfies an important technical condition.  Then the MC~ribbon and the HC~ribbon belong to a family of measures of correlations all of which satisfy tensorization.  

The technical condition that we require $\Phi$ to satisfy is exactly the same condition under which $\Phi$-entropy becomes \emph{subadditive}. Subadditivity of entropy for \emph{independent} random variable is a tool that is used to prove certain concentration of measure inequalities. Our $\Phi$-ribbon, defined for arbitrary \emph{correlated} random variables, can be understood as a generalization of the subadditivity inequality of $\Phi$-entropy.

Studying $\Phi$-ribbon further, we show that a quantity introduced in~\cite{Raginsky}, called the \emph{strong data processing inequality constant}, can be characterized in terms of $\Phi$-ribbon in the same way that the MC~ribbon characterizes $\rho$. Moreover, we show that the MC~ribbon, as a set, includes all other $\Phi$-ribbons and in this sense is a special one. Moreover, we prove equivalent characterizations for the MC~ribbon which help us to compute it more easily. In particular, we compute the MC~ribbon of a multivariate Gaussian distribution in terms of its covariance matrix. We also
fully characterize the MC~ribbon in the bipartite case in terms of maximal correlation.

%*************************************

\section{Preliminaries}

Let us first fix some notations. Sets are denoted by calligraphic letters as $\mathcal X$. Random variables are denoted by capital letters as $X$ and their values by lowercase letters as $x\in \mathcal X$. Such a random variable is determined by its distribution $p_X$, i.e., with values $p(X=x)=p(x)$ for $x\in \mathcal X$. Except otherwise stated, we restrict to random variables taking values in finite sets.

We let $[k]=\{1,2,\dots, k\}$. The tuple $(\lambda_1, \lambda_2, \dots, \lambda_k)$ is sometimes denoted by $\lambda_{[k]}$. Similarly, when we have $k$ random variables $X_1, \dots, X_k$, we use $X_{[k]}$ to denote the tuple $(X_1, X_2, \dots, X_k)$ for $k\geq 1$. When $k=0$, we use $X_{[k]}$ to denote the empty sequence. 
We also use $\widehat i$ to denote $\{1, \dots, i-1, i+1, \dots, k\}$, so $X_{\widehat i} = (X_1, \dots, X_{i-1}, X_{i+1}, \dots, X_k)$.

Let $X$ be a random variable taking values in the finite set $\mathcal X$. Then a function $f:\mathcal X\rightarrow \mathbb R$ can itself be thought of as a random variable. To emphasis that $f$ is a function of $X$ we sometimes denoted it by $f_X$ or $f(X)$. The expectation and variance of $f$ are denoted by $\E[f]$ and $\Var[f]$ respectively. We sometimes denoted them by $\E_X[f]$ and $\Var_X[f]$ to emphasis that they are computed with respect to the random choice of $X$. 

Let $f=f_{XY} = f(X, Y)$ be a function of two random variables $(X, Y)$ with the joint distribution $p_{XY}$. Then $\E[f|X]$ is a function of $X$ which is equal to the conditional expectation of $f$, over the random choice of $Y$, given a fixed value for $X$:
$$\E[f|X](x) = \E[f|X=x] = \sum_y p(y|x) f(x, y).$$
We can then speak of $\Var [\E[f|X]]=\Var_X\big[\E[f|X]\big]$.

A function $\Phi$ is said to be smooth if it has derivatives of all orders everywhere in its domain.

We denote the binary entropy function by $h(\cdot)$, i.e., $h(p)=-p\log p - (1-p)\log(1-p)$ for $p\in [0,1]$.

\subsection{Hypercontractivity ribbon}

We have already defined an important measure of correlation with the tensorization property in~\eqref{eq:max-correlatoin-31}. Another important such measure is the \emph{hypercontractivity ribbon} first defined by Ahlswede and G\'acs~\cite{AhlswedeGacs}. 

\begin{definition}[\cite{AhlswedeGacs}] 
The hypercontractivity ribbon (HC~ribbon), $\fR(X, Y)$, associated to a pair of random variables $(X, Y)$ is the set of all $(\lambda_1, \lambda_2)\in [0,1]^2$ such that for every pair of functions $f_{X}$ and $g_{Y}$ we have
\begin{align}\label{eq:fg-norm-2}
\E[f_{X}g_{Y}]\leq \big\|f_{X}\big\|_{\frac{1}{\lambda_1}}\big\|g_{Y}\big\|_{\frac{1}{\lambda_2}},
\end{align}
where the norms $\|\cdot\|_r$ are defined by $\|f\|_{r}=\E\big[|f|^{r}\big]^{1/r}$.
\end{definition}

We should mention here that the HC~ribbon defined in~\cite{AhlswedeGacs} is indeed the set of $(r, s)=(1/\lambda_1, 1/\lambda_2)$ for which $(\lambda_1, \lambda_2)\in \fR(X, Y)$ as we defined above. Nevertheless, we prefer this definition for later use. 

HC~ribbon satisfies several interesting properties for which we refer to~\cite{AhlswedeGacs}. Here we only mention the surprising result of Nair~\cite{Nair} that HC~ribbon can be characterized in terms of mutual information (a related characterization was also found in~\cite{NairPre}).

\begin{theorem}[\cite{Nair}] \label{thm:Nair}
$\fR(X, Y)$ is equal to the set of all pairs $(\lambda_1, \lambda_2)\in [0,1]^2$ such that for all $p(u| x, y)$ we have
\begin{align} \label{eqn:Ubin-pre}
I(XY;U)\geq \lambda_1 I(X;U)+\lambda_2 I(Y;U).
\end{align}
Furthermore, without loss of generality one may restrict to auxiliary random variables $U$ that are binary.
\end{theorem}

An important quantity related to HC~ribbon is the \emph{strong data processing inequality} constant. We refer to~\cite{AhlswedeGacs} for its original definition. Here, based on the result of~\cite{AGKN}, we may define this constant $s^*(X, Y)$, as the smallest $\lambda\geq 0$ such that for any $p(u|x)$ we have
$$\lambda I(U; X)\geq I(U; Y).$$
Note that for a Markov chain $U-X-Y$, by the data processing inequality we have $I(U;X)\geq I(U;Y)$ (and then $s^*(X, Y)\leq 1$). That is the reason that $s^*(X, Y)$ is called the \emph{strong} data processing inequality constant. 

$s^*(X, Y)$ can be characterized in terms of HC~ribbon as follows:
$$s^*(X, Y) = \inf \frac{1-\lambda_1}{\lambda_2},$$
where the infimum is taken over all $(\lambda_1, \lambda_2)\in \fR(X, Y)$. Observe that this characterization of $s^*(X, Y)$ is similar to that of $\rho(X, Y)$ given in~\eqref{eq:rho2-mc-ribbon}.

%\begin{theorem}[\cite{OurPaper}] Given any $p(x_1, x_2)$, we have
%\begin{align}\label{eq:rho-ribbon}
%\rho^2(X_1, X_2) = \inf  \frac{1-\lambda_1}{\lambda_2},
%\end{align}
%where the infimum is taken over all $(\lambda_1, \lambda_2)\in \fS(X_1, X_2)$.
%\end{theorem}

It is straightforward to generalize the definition of HC~ribbon as well as Theorem~\ref{thm:Nair} to the multivariate case. The HC~ribbon, $\fR(X_1, \dots, X_k)$, associated to $k$ random variables $X_{[k]}=(X_1,\dots, X_k)$, is the set of tuples $(\lambda_1, \dots, \lambda_k)\in [0,1]^k$ such that for all functions $f_i=f_i(X_i)$, $i=1, \dots, k$, we have
$$\E[f_1\cdots f_k]\leq \big\|f_1\big\|_{\frac{1}{\lambda_1}}\cdots \big\|f_k\big\|_{\frac{1}{\lambda_k}}.$$
Then it is easily verified that Theorem~\ref{thm:Nair}, with the same proof, holds in the multivariate case. That is, $\fR(X_1, \dots, X_k)$  is equal to the set of tuples $(\lambda_1, \dots, \lambda_k)\in [0,1]^k$ such that for any auxiliary (binary) random variable $U$ we have 
\begin{align} \label{eqn:Ubin}
I(X_{[k]};U)\geq \lambda_1 I(X_1;U) +\cdots + \lambda_k I(X_k;U).
\end{align}

%\begin{definition}\label{def-func-f}
%Let $\fS(X_{[k]})$ be the set of  $k$-tuples $(\lambda_1, \lambda_2, \cdots, \lambda_k)\in[0,1]^k$ such that
%\begin{align}\Var[f(X_{[k]})]\geq \sum_{i=1}^{k}\lambda_i\Var_{X_i}[\mathbb{E}[f(X_{[k]})|X_i]]\label{MCribbonDef13}
%\end{align}
%for all  functions $f:\mathcal{X}_{[k]}\mapsto \mathbb{R}$. 
%\end{definition}

%\begin{theorem}
%\label{myRepeatedTheorem}
%The MC~ribbon can be expressed as the set of the $k$-tuple region $(\lambda_1, \lambda_2, \cdots, \lambda_k)\in[0,1]^k$ such that for any functions $f_i:\mathcal{X}_i\mapsto \mathbb{R}$ for $i\in[k]$, we have
%\begin{align}\Var[\sum_{i=1}^k f_i]\leq \sum_{i=1}^{k}\frac{1}{\lambda_i}\Var[f_i].\label{eqn:fc2def}\end{align}
%\end{theorem}

\subsection{$\Phi$-entropy}

To present our main results we need to define and review the properties of $\Phi$-entropy. 
The reader may refer to \cite[Chapter 14]{Boucheronetal} for a more detailed treatment of the subject (see also \cite{Chafai2, Chafai1}).

Let $f=f_X$ be a function of a random variable $X$. Also let $\Phi$ be a function that is defined on a convex set that contains the range of $f$. Then the $\Phi$-entropy of $f$ is defined by
$$H_\Phi(f) = \E[\Phi(f)] - \Phi(\E f).$$
In this paper we always assume that $\Phi$ is convex, in which case 
$$H_\Phi(f)\geq 0,$$ 
by Jensen's inequality. For the choice of $\Phi(t)=t^2$, the $\Phi$-entropy simply reduces to variance: $H_\Phi(f)=\Var(f)$.

\begin{example} \label{ex:1}
Let $p_{UA}$ be some arbitrary distribution with $U$ taking values in $\{+1, -1\}$.  Define $f_A = \E[U|A]$. Then $\E[f] = \E[U]$ and we have
\begin{align*}
I(U; A) & = H(U) - H(U|A) \\
& = h\Big(\frac{1+ \E f}{2}\Big)  - \E\Big[  h\Big(  \frac{1+f}{2}   \Big) \Big]\\
&= H_\Phi(f),
\end{align*}
for $\Phi(x) = 1- h(\frac{1+x}{2})$, where $h(\cdot)$ denotes the binary entropy function.
\end{example}

Similar to the conditional Shannon entropy, we define the conditional $\Phi$-entropy. Let $f_{XY}$ be a function of two random variables $(X,Y)$. Then we define 
\begin{align}H_\Phi(f|Y)& =\E[\Phi(f)]-\mathbb{E}_Y[\Phi(\E[f|Y])]\\
 &= \sum_{y}p(y)\big(\E[\Phi(f)|Y=y] - \Phi(\E [f|Y=y])\big).
 \end{align}
Furthermore, we set $H_\Phi(f|Y=y) = \E[\Phi(f)|Y=y] - \Phi(\E [f|Y=y])$, so that we have
$$H_\Phi(f|Y) = \sum_{y}p(y)H_\Phi(f|Y=y).$$
With these notations, we can now express $\Phi$-entropy's version of the \emph{law of total variance}: 
\begin{align}
H_\Phi(f)&=\E[\Phi(f)] -\Phi(\E f)\nonumber
\\&=\E[\Phi(f)]-\mathbb{E}_{Y}\Phi(\E[f|Y])+\mathbb{E}_{Y}[\Phi(\E[f|Y])] -\Phi(\E f)\nonumber
\\&=H_{\Phi}(f|Y)+H_\Phi(\E[f|Y]).
\label{law-of-tot}
\end{align}
We call the above equation \emph{the chain rule} for $\Phi$-entropy as it parallels the chain rule for Shannon entropy. 

Along the same lines, one can prove the following conditional form of the chain rule for $\Phi$-entropy. Suppose that $f_{XYZ}$ is a function of three random variables $(X, Y, Z)$. Then we have
\begin{align}\label{eq:chain-rule-Phi}
H_\Phi(f|X) = H_\Phi(f|XY) + H_\Phi\big(  \E[f|XY]   |X \big),
\end{align}
which is a generalization of~\eqref{law-of-tot}.
This equation and the non-negativity of $\Phi$-entropy imply that 
\begin{align}
H_\Phi(f|X)\geq H_{\Phi}(f|XY).\label{conditioning-Hphi}
\end{align}
In other words, just like Shannon's entropy, conditioning reduces $\Phi$-entropy.   Observe that from the chain rule, \eqref{conditioning-Hphi} can be also written as
\begin{align}
H_{\Phi}(\mathbb{E}[f|XY])\geq H_\Phi(\mathbb{E}[f|X]).\label{conditioning-Hphi-new}
\end{align}

Despite the above similarities between Shannon's entropy and the $\Phi$-entropy, one can relate  $\Phi$-entropy to the generalized \emph{relative entropy} of Ali-Silvey \cite{Ali-Silvey} and Csiszar \cite{Csiszar1}\cite{Csiszar2} (also called the ``$f$-divergence"): take a non-negative function $f(x)$ that is normalized, \emph{i.e.,} $\mathbb{E}_X[f]=1$. Then $f(x)$ is of the form $f(x)=q(x)/p(x)$ where $p(x)$ is the underlying distribution on $X$ and $q(x)$ is some arbitrary distribution. Now,
\begin{align}H_{\Phi}(X)&=\sum_{x}p(x)\Phi(f(x)) -\Phi(\sum_x p(x)f(x))
\\&=\sum_{x}p(x)\Phi(\frac{q(x)}{p(x)}) -\Phi(1) \label{eqn:Hphi-relative-ent}
\end{align}
is explicitly in terms of the $\Phi$-divergence. Ignoring the $\Phi(1)$ term, the $\Phi$-entropy in \eqref{eqn:Hphi-relative-ent} reduces to the KL divergence $D(q\|p)$ for  $\Phi(x)=x\log(x)$. Therefore, $\Phi$-entropy is really a relative entropy (when $f$ is a non-negative and normalized) rather than an entropy. In fact, the analogy between Shannon's entropy and the $\Phi$-entropy has limitations: while Shannon entropy is \emph{concave} in its underlying distribution, the following \emph{convexity} property holds for the $\Phi$-entropy.

% to do: application section
% f divergence
% strong subadditivity
%Yes, as I told you before Phi-entropy is really a relative entropy, rather than an entropy. Take f to be normalized: Ef=1 (and positive). Then f is of the form f=q/p where p is the underlying distribution and q is some arbitrary distribution. Now the Phi-entropy of f, is related to the Phi-divergence of q relative to p. This is the case at least for the usual choices of Phi.

%Just like Shannon's entropy function, the following  convexity property holds for the $\Phi$-entropy.

\begin{lemma}\label{lem:convexity-H-Phi-channel}
Let $\Phi$ be a convex function and fix the distribution $p_X$ and function $f_X$. Then the function
$$p_{Y|X}\mapsto H_\Phi(\E[f|Y]),$$
is convex.
\end{lemma}

\begin{proof}
We note that $H_\Phi(\E[f|Y]) = \E_Y[\Phi(\E[f|Y])] - \Phi(\E f)$. So it suffices to prove the convexity of
$$p_{Y|X}\mapsto \E_Y[\Phi(\E[f|Y])] = \sum_y p(y) \Phi\Big( \sum_{x}   \frac{p(x) p(y|x) f(x)}{p(y)}   \Big),$$
which is immediate once we note that for any convex $\Phi$, the function 
$$(s, t)\mapsto s\Phi\big(\frac{t}{s}\big),$$
 is jointly convex for $s>0$.
\end{proof}

The following simple lemma will be used frequently.

\begin{lemma}\label{eq:phi-entropy-taylor}
Let $\Phi$ be a smooth convex function. Let $c\in\mathbb{R}$ be point in the interior of the domain of $\Phi$, and $f$ be an arbitrary function with $\E f=0$. Then for $g=c+\epsilon f$, where $|\epsilon|$ is small, we have
$$H_\Phi(g) = \frac{1}{2}\Phi''(c)\Var[f]\epsilon^2 + O(|\epsilon|^3).$$
\end{lemma}

\begin{proof}
Taking the Taylor expansion of $\Phi$ around $c$ we have 
$$\Phi(c+\epsilon f) -\Phi(c) = \Phi'(c)(\epsilon f)+\frac{1}{2}\Phi''(c)(\epsilon f)^2 + O(|\epsilon|^3).$$
Now taking the expectation of both sides and noting that $\E g=c$, we obtain the desired result. 
\end{proof}

So far the only condition we put on $\Phi$ is convexity. We must however consider a more restricted class of functions. 

\begin{definition}\label{def:class-C}
We define $\mathscr F$  to be the class of smooth convex functions $\Phi$,  whose domain is a convex subset of $\mathbb{R}$, that are \emph{not} affine (not of the form $ at+b$ for some constants $a$ and $b$) and satisfy one of the following \emph{equivalent} conditions (see \cite[Exercise 14.2]{Boucheronetal}):
\begin{enumerate}[label={\rm (\roman*)}]
\item  $(s, t)\mapsto p \Phi(s) + (1-p) \Phi(t) - \Phi(p s+ (1-p) t) $, for any $p\in[0,1]$, is jointly convex.
\item  $(s, t)\mapsto \Phi(s) - \Phi(t) - \Phi'(t)(s-t)$ is jointly convex.
\item $(s, t)\mapsto (\Phi'(s)- \Phi'(t))(s-t)$ is jointly convex.
\item $(s, t)\mapsto \Phi''(s)t^2$ is jointly convex.
\item $1/\Phi''$ is concave.
\item  $\Phi'''' \Phi''\geq 2 \Phi'''^2$.
\end{enumerate}
\end{definition}

Let us clarify a few points in this definition. First, we exclude affine functions $\Phi(t)=at+b$ simply because $H_\Phi(f)$ always vanishes if $\Phi$ is affine. Second, from the above list of equivalent conditions, we mostly use (i) which has nothing to do with the smoothness of $\Phi$. We indeed assumed smoothness only because in this case we have the equivalent conditions (v) and (vi) which can easily be verified.  Third, using (v)
$\Phi''(x)$ is strictly positive for any $\Phi\in \mathscr F$. That is, functions in $\mathscr F$ are strictly convex.  

Examples of functions in $\mathscr F$ include $\Phi(t)=t\log t$ and $\Phi(t)=t^\alpha$ for $\alpha\in(1, 2]$ as well as their affine transformations such as $\Phi(t)=1-h(\frac{1+t}{2})$ and $\Phi(t)=(1-t)^\alpha+(1+t)^\alpha$ for $\alpha\in(1, 2]$.

The following lemma is a key tool in our proofs in the next section. 

\begin{lemma}\label{lem:key-lemma-phi-entropy} 
\begin{enumerate}[label={\rm{(\alph*)}}]
\item Assume $X$ and $Y$ are \emph{independent} random variables, and $f_{XY}$ is arbitrary. Then, for any $\Phi\in\mathscr F$, we have
$$\E[\Phi(f)]-\mathbb{E}_{X}[\Phi(\E_Y[f|X])]
\geq \E_Y[\Phi(\E_X[f|Y])] -\Phi(\E f),$$
or equivalently
$H_{\Phi}(f|X)\geq H_\Phi(\E[f|Y])$. 
\item More generally, if $f_{XYZ}$ is a function of three random variables satisfying the Markov chain condition $X-Z-Y$, we have
$$H_\Phi(f|XZ)\geq H_{\Phi}\big(\E[f|YZ]\big|Z\big).$$
\item Under the same condition as in part (b) we have
$$H_{\Phi}(\E [f|Z]) + H_{\Phi}(f|XZ) \geq H_\Phi(\E [f|YZ]).$$

\end{enumerate}

\end{lemma}

\begin{proof}
(a) Based on property (i) of Definition~\ref{def:class-C},  an induction argument shows that 
for every distribution $p_X$, the mapping
$$f_X\mapsto \sum_x p(x) \Phi(f(x)) - \Phi\Big(\sum_x p(x) f(x)\Big),$$
is jointly convex. This means that for every  distribution $q_Y$ and $f_{XY}$ we have 
\begin{align*}
\sum_y q(y)  \Bigg(\sum_x p(x) \Phi(f(x, y)) - \Phi\big(&\sum_x p(x) f(x, y)\big)\Bigg)\\
& \geq 
 \sum_x p(x) \Phi\Big( \sum_y q(y) f(x, y)\Big) - \Phi\Big(\sum_{x, y} p(x) q(y) f(x, y)\Big).
\end{align*}
This is equivalent to $H_{\Phi}(f|X)\geq H_\Phi(\E[f|Y])$. 

\vspace{.2in} \noindent
(b) This part is just the ``conditional" version of (a). To prove this, write down the inequality of part (a) for the function $g_{XY}^{(z)}(x, y) = f(x, y, z)$, for every fixed $Z=z$, and then take average over $z$.  
%This is a trick that we use several times in this section.
Moreover, (c) follows form (b) once we use the chain rule $H_{\Phi}\big(\E[f|YZ]\big|Z\big) = H_{\Phi}(\E[f|YZ]) - H_\Phi(\E[f|Z])$.

\end{proof}

Subadditivity is a desirable property of $\Phi$-entropy \cite[Sec. 4.13]{Boucheronetal}. Let us now prove the subadditivity of $\Phi$-entropy as a corollary of Lemma~\ref{lem:key-lemma-phi-entropy}.

\begin{theorem}[Subadditivity of $\Phi$-entropy] 
Let $(X_1, \dots, X_k)$ be \emph{mutually independent} random variables and $f$ be an arbitrary function of them. Then for any $\Phi\in\mathscr F$, we have
$$H_\Phi(f)\leq \sum_{i=1}^k H_\Phi\big(f| X_{\widehat i}\big),$$
where $X_{\widehat i} = (X_1, \dots, X_{i-1}, X_{i+1}, \dots, X_k)$.
\end{theorem}

\begin{proof}
Recall that $X_{[j]} = (X_1, \dots, X_j)$. Using the conditional form of the chain rule~\eqref{eq:chain-rule-Phi} as well as part (b) of Lemma~\ref{lem:key-lemma-phi-entropy}, for every $0\leq j\leq k-1$ we have
\begin{align*}
H_\Phi(f|X_{[j]}) & = H_\Phi\big(f| X_{\widehat{j+1}}\big) + H_\Phi\big( \E[f|X_{\widehat{j+1}}] |X_{[j]}  \big)\\
& \leq H_\Phi\big(f| X_{\widehat{j+1}}\big) + H_\Phi(f| X_{[j+1]}).
\end{align*}
Summing up all these inequalities gives the desired result. 
\end{proof}

Observe that subadditivity of $\Phi$-entropy for the choice of $\Phi(t)=t^2$ is nothing but the Efron-Stein inequality.  Using the chain rule for $\Phi$-entropy \eqref{law-of-tot},
$$H_\Phi(f)=H_\Phi\big(f| X_{\widehat i}\big)+H_\Phi\big(\mathbb{E}[f| X_{\widehat i}]\big)$$ we can equivalently express the sub-additivity of $\Phi$-entropy as
\begin{align}H_\Phi(f)\geq \sum_{i=1}^k \frac{1}{k-1}H_\Phi\big(\mathbb{E}[f| X_{\widehat i}]\big).\label{eqn:subadditlnf}\end{align}
From here, it is a short trip to motivate our notion of $\Phi$-ribbon, formally defined in the next section and studied in the rest of this paper. 
Let us ask for the set of possible non-negative coefficients $\lambda_i$ for which
$$H_\Phi(f)\geq \sum_{i=1}^k \lambda_i H_\Phi\big(\mathbb{E}[f| X_{\widehat i}]\big)$$
holds for all functions $f$, \emph{i.e.,} we are asking for the best possible constants that one can substitute instead of $1/(k-1)$ in \eqref{eqn:subadditlnf}. This question about the set of coefficients $\lambda_i$ can be asked even when $X_1, X_2, \dots, X_k$ are \emph{correlated} sources. Letting $Y_i=X_{\widehat i}$, we can think of $f$ as a function of $(Y_1, Y_2, \cdots, Y_k)$, and ask for the set of coefficients $\lambda_i$ such that 
$$H_\Phi(f)\geq \sum_{i=1}^k \lambda_i H_\Phi\big(\mathbb{E}[f| Y_i]\big)$$
for all functions $f$ of $(Y_1, Y_2, \cdots, Y_k)$. This is what we call the $\Phi$-ribbon associated to $(Y_1, Y_2, \cdots, Y_k)$.

%**************************************************************

\section{$\Phi$-ribbon}\label{secPhiRib}

In this section we present our main definition, namely the $\Phi$-ribbon, and prove some of its properties. In particular, we show that it generalizes both the MC and the HC~ribbons, and satisfies the tensorization  and monotonicity properties.

\begin{definition}
Let $\Phi\in \mathscr F$. 
For arbitrarily distributed random variables $(X_1, X_2, \cdots, X_k)$ we define its $\Phi$-ribbon, denoted by $\fR_\Phi(X_1, \dots, X_k)=\fR(X_{[k]})$, to be the set of all $k$-tuples $(\lambda_1, \lambda_2, \cdots, \lambda_k)$ of non-negative numbers such that for every function $f_{X_{[k]}}$ we have
\begin{align}
H_\Phi(f)\geq \sum_{i=1}^k\lambda_i H_\Phi(\E[f|X_i]).
\label{def:strong-sub}
\end{align}
 Note that we require the above equation for any function $f_{X_{[k]}}$ whose range is in the domain of $\Phi\in \mathscr F$.
\end{definition}

From the definition it is clear that $\Phi$-ribbon for the choice of $\Phi(t)=t^2$ reduces to the MC~ribbon defined in the Introduction. Furthermore, according to Example~\ref{ex:1} and Theorem~\ref{thm:Nair}, if $\Phi(t) = 1- h(\frac{1+t}{2})$, the ribbon $\fR_\Phi(X_{[k]})$ becomes the HC~ribbon.

Letting $f$ to be only a function of $X_i$, for some $i\in [k]$, we find that $\E[f|X_i]=f$ and hence $H_\Phi(\E[f|X_i]) = H_\Phi(f)$. Then, for any $(\lambda_1, \dots, \lambda _k)\in \fR_\Phi(X_1, \dots, X_k)$ we must have $\lambda_i\leq 1$; that is, 
$$\fR_\Phi(X_1, \dots, X_k)\subseteq [0,1]^k.$$ 
On the other hand, using the chain rule and the fact that $\Phi$-entropy is non-negative, we have $H_\Phi(f)\geq H_\Phi(\E[f|X_i])$ for every $i$. Therefore, the tuple $(\lambda_1, \dots, \lambda_k)$ of non-negative numbers, always belongs to $\fR_\Phi(X_{[k]})$ if $\sum_i \lambda_i\leq 1$. This means that the nontrivial part of the $\Phi$-ribbon is the set of $k$-tuples of non-negative numbers whose sum is greater than one. 

\begin{example} 
If $X_{1}=X_2=\dots=X_k$ are non-constant, then $\E[f|X_i]=f$ and $H_\Phi(f)= H_\Phi(\E[f|X_i])$, for every $i$. As a result, 
$\fR_\Phi(X_{[k]})$ contains only those $(\lambda_1, \lambda_2, \dots, \lambda_k)\in [0,1]^k$ that satisfy  $\sum_i\lambda_i\leq 1$.
\end{example}

In the following we show that the $\Phi$-ribbon is the largest possible ribbon when $X_i$'s are mutually independent. 

\begin{proposition} \label{prop:indep-phi-ribbon}
For any $\Phi\in \mathscr F$, if $(X_1, \dots, X_k)$ are mutually independent, we have $\fR_\Phi(X_{[k]})=[0,1]^k.$ 
\end{proposition}

\begin{proof}
We need to show that 
\begin{align}\label{eq:independent-phi-ineq}
H_\Phi(f)\geq \sum_{i=1}^kH_\Phi(\E[f|X_i]).
\end{align}
For this, we again use the conditional form of the chain rule~\eqref{eq:chain-rule-Phi} as well as part (b) of Lemma~\ref{lem:key-lemma-phi-entropy}. For any $0\leq j\leq k-1$ we have
\begin{align*}
H_\Phi(f|X_{[j]}) & = H_\Phi(f| X_{[j+1]}) + H_\Phi\big( \E[f| X_{[j+1]}] | X_{[j]}  \big)\\
& \geq H_\Phi(f| X_{[j+1]}) + H_\Phi( \E[f| X_{j+1}] \big),
\end{align*}
where in the second line we use Lemma~\ref{lem:key-lemma-phi-entropy} for the function $\E[f|X_{[j+1]}]$. Summing up all these inequalities gives the desired result. 
\end{proof}

We can now prove the main result of this section, namely the tensorization and monotonicity properties of the HC~ribbon and MC~ribbon extend to the $\Phi$-ribbon.

\begin{theorem}\label{data-process-tens-Phi}
For any $\Phi\in \mathscr F$, the $\Phi$-ribbon satisfies monotonicity and tensorization as follows:
\begin{enumerate}[label=\rm{(\roman*)}]
\item Data processing: if $(X_{[k]}, Y_{[k]})$ are random variables whose joint probability distribution satisfies $p(y_{[k]}|x_{[k]})=\prod_{i=1}^np(y_i|x_i)$, then 
$$\fR_\Phi(X_{[k]})\subseteq \fR_\Phi(Y_{[k]}).$$
\item Tensorization: if $X_{[k]}$ are independent of $Y_{[k]}$, i.e.,  $p(x_{[k]}, y_{[k]})=p(x_{[k]})p(y_{[k]})$, then $$\fR_\Phi(X_1Y_1, X_2Y_2, \dots, X_kY_k ) = \fR_\Phi(X_1, X_2, \dots, X_k)\cap \fR_\Phi(Y_1, Y_2, \dots, Y_k).$$
\end{enumerate}
\end{theorem}

\begin{proof}
(i) Let $\lambda_{[k]}\in \fR_\Phi(X_{[k]})$, we show that  $\lambda_{[k]}\in \fR_\Phi(Y_{[k]})$. For this we need to show that for every function $f_{Y_{[k]}}$ of $Y_{[k]}$ we have 
$$H_\Phi(f)\geq \sum_i \lambda_i H_\Phi(\E[f|Y_i]).$$  

Using the definition of $\lambda_{[k]}\in \fR_\Phi(X_{[k]})$ applied to the function $\E[f|X_{[k]}]$ we find that
\begin{align}H_\Phi(\E[f|X_{[k]}]) \geq \sum_{i=1}^k \lambda_i H_{\Phi}(\E[f| X_i]) .\label{neweqader1}\end{align}
On the other hand, since $Y_{[k]}$ are mutually independent conditioned on $X_{[k]}$, using Proposition~\ref{prop:indep-phi-ribbon} (in fact the ``conditional" version of~\eqref{eq:independent-phi-ineq}) we find that    
\begin{align}
H_\Phi(f|X_{[k]}) \geq \sum_{i=1}^k H_\Phi( \E[f|X_{[k]}Y_i]  |X_{[k]})\geq \sum_{i=1}^k \lambda_i H_\Phi( \E[f|X_{[k]}Y_i]  |X_{[k]}).\label{neweqader2}
\end{align}
Summing up \eqref{neweqader1} and \eqref{neweqader2}, and using the chain rule we arrive at
$$H_\Phi(f)\geq \sum_{i=1}^k \lambda_i \Big(    H_{\Phi}(\E[f| X_i])+    H_\Phi( \E[f|X_{[k]}Y_i]  |X_{[k]})  \Big).$$
Therefore, it suffices to verify that $H_{\Phi}(\E[f| X_i])+    H_\Phi( \E[f|X_{[k]}Y_i]  |X_{[k]})\geq H_\Phi(\E[f|Y_i])$ for every $i$. 
For a fixed $i$, let $g_{X_{[k]}Y_i}= \E[f| X_{[k]}Y_i]$. Then this inequality can be written as 
\begin{align}\label{eq:45-0}
H_{\Phi}(\E[g|X_i]) + H_{\Phi}(g|X_{[k]}) \geq H_\Phi(\E[g|Y_i]).
\end{align}
Now note that we have the Markov chain $X_{\widehat i}- X_i -Y_i$, so by part (c) of Lemma~\ref{lem:key-lemma-phi-entropy} we have 
$$H_{\Phi}(\E[g|X_i]) + H_{\Phi}(g|X_{[k]}) \geq H_\Phi(\E[g|X_iY_i]).$$
Then~\eqref{eq:45-0} follows from~\eqref{conditioning-Hphi-new}.

%%%%%
\vspace{.25in}
\noindent
(ii) In the definition of $\fR_\Phi(X_1Y_1, X_2Y_2, \dots, X_kY_k)$ by restricting to functions $f_{X_{[k]}}$ of $X_{[k]}$ only, or to functions $f_{Y_{[k]}}$ of $Y_{[k]}$ only, we find that $$\fR_\Phi(X_1Y_1, X_2Y_2, \dots, X_kY_k) \subseteq \fR_\Phi(X_1, X_2, \dots, X_k)\cap \fR_\Phi(Y_1, Y_2, \dots, Y_k).$$

To prove the inclusion in the other direction, let $$(\lambda_1, \dots, \lambda_k)\in \fR_\Phi(X_1, X_2, \dots, X_k)\cap \fR_\Phi(Y_1, Y_2, \dots, Y_k),$$ and let $f_{X_{[k]}Y_{[k]}}$ be arbitrary. We need to show that 
\begin{align}
H_\Phi(f)\geq \sum_{i=1}^k \lambda_i H_\Phi(\E[f| X_iY_i]).
\label{eqnnttshow}
\end{align}

Using our assumption on $(\lambda_1, \dots, \lambda_k)$ for function $\E[f|X_{[k]}]$, as a function of $X_{[k]}$, we have
$$H_\Phi(\E[f| X_{[k]}]) \geq \sum_{i=1}^k \lambda_i H_\Phi(\E[f|X_i]).$$
Next considering $f$, for a fixed $X_{[k]}=x_{[k]}$, as a function of $Y_{[k]}$. Observe that conditioned on $X_{[k]}=x_{[k]}$ the distribution of $Y_{[k]}$ does not change. Therefore, since $(\lambda_1, \dots, \lambda_k)$ belongs to $\fR(Y_1, \dots, Y_k)$ we have 
$$H_{\Phi}(f|X_{[k]}) \geq \sum_{i=1}^k \lambda_i H_\Phi\big(\E[f|X_{[k]} Y_i]\big| X_{[k]}\big).$$
Summing up these two inequalities and using chain rule,~\eqref{eqnnttshow} is implied if we have 
$$H_\Phi(\E[f|X_i]) + H_\Phi\big(\E[f|X_{[k]} Y_i] \big| X_{[k]}\big)\geq H_\Phi(\E[f| X_iY_i]),\qquad \forall i.$$
This inequality follows once we note that we have the Markov chain $X_{\widehat i} - X_i- Y_i$, and we can use part (c) of Lemma~\ref{lem:key-lemma-phi-entropy} for the function $\E[f|X_{[k]}Y_i]$.

\end{proof}

Theorem \ref{thm:Nair} provides a description of the HC ribbon in terms of mutual information. Given $p(x,y)$, one can define the $\Phi$-mutual information between $X$ and $Y$ as follows \cite{Raginsky}:
$$I_{\Phi}(X;Y)=\sum_{x,y}p(x)p(y)\Phi(\frac{p(x,y)}{p(x)p(y)}) - \Phi(1).$$
This definition differs from the one used in \cite{Raginsky} due to the subtraction of the $\Phi(1)$ term. This subtraction ensures that $I_\Phi(X;Y)=0$ when $X$ and $Y$ are independent. Observe that $\Phi$-mutual information reduces  to Shannon's mutual information for $\Phi(x)=x\log(x)\in  \mathscr F$. Then, similar to the statement of Theorem \ref{thm:Nair}, we may define the $I_\Phi$-ribbon as follows:

\begin{definition}
Let $\Phi\in \mathscr F$. 
For arbitrarily distributed random variables $(X_1, X_2, \cdots, X_k)$ we define its $I_{\Phi}$-ribbon, denoted by $\fR_{I_\Phi}(X_{[k]})$, to be the set of all $k$-tuples $(\lambda_1, \lambda_2, \cdots, \lambda_k)$ of non-negative numbers such that for any $p(u|x_{[k]})$, we have
$$\sum_i\lambda_iI_{\Phi}(U;X_i)\leq  I_{\Phi}(U;X_{[k]}).$$
\end{definition}
Then, we have the following:
\begin{theorem}\label{thm:newamiphi}We have
$\fR_{I_\Phi}(X_{[k]})=\fR'_{\Phi}(X_{[k]})$
where $\fR'_{\Phi}$ is the set of $(\lambda_1, \dots, \lambda_k)$ such that  for any  $f(x_{[k]})$ satisfying $f\geq 0$ and $\mathbb{E}[f]=1$,  we have
$$\sum_i\lambda_iH_{\Phi}(\mathbb{E}[f|X_i])\leq  H_{\Phi}(f).$$
Observe that  $\fR'_{\Phi}$ has the same definition as  $\fR_{\Phi}$, except that in $\fR'_{\Phi}$ we restrict to functions $f(x_{[k]})$ satisfying $f\geq 0$ and $\mathbb{E}[f]=1$.
\end{theorem}
The condition $f\geq 0$ and $\mathbb{E}[f]=1$ essentially says that $f(x_{[k]})$ can be written as $q(x_{[k]})/p(x_{[k]})$ where $p(x_{[k]})$ is the given distribution on $X_{[k]}$ and $q(x_{[k]})$ is some arbitrary distribution. The proof of this theorem is given in Appendix \ref{app:thm:newamiphi}.

\subsection{Examples}\label{subsection:Phi-examples}
Natural choices for the function $\Phi(t)\in \mathscr F$ are  
$$\varphi_\alpha(t)= t^\alpha, \qquad \alpha\in (1,2].$$
Note that $\varphi(x)$ is defined only for $t\geq 0$. Without changing the corresponding $\Phi$-ribbon, we can even restrict the domain of $\varphi_\alpha$ to $[0,1]$ simply because $\varphi_{\alpha}(ct) = c^{\alpha}\varphi_\alpha(t)$ implies that the equation
$
H_\Phi(f)\geq \sum_{i=1}^k\lambda_i H_\Phi(\E[f|X_i])
$
holds for $f$ if and only if it holds for a scaled version of $f$.

 Another special choice of interest is $\Phi_1(t)=1-h\big(\frac{1+t}{2}\big)$ which, as mentioned before, results in the HC~ribbon. One way to understand this function is to consider the class of functions $\Phi_\alpha$ for $\alpha\in (1,2]$ defined by
\begin{align}
\Phi_\alpha(t) = \frac{(1+t)^\alpha + (1-t)^\alpha -2}{2^\alpha-2}, \qquad & \alpha\in (1, 2],
\quad t\in[-1,1].
\end{align}
Then $\Phi_\alpha\in \mathscr F$, and we have 
$$\lim_{\alpha \searrow 1} \Phi_\alpha(t) = \Phi_1(t), \qquad \forall t\in[-1,1].$$

Note that $\Phi_2=\varphi_2$ for which the associated $\Phi$-ribbon is equal to the MC~ribbon.

%We can similarly use $\fR_{\Phi_\alpha}(X_{[k]})$ to denote the corresponding $\Phi$-ribbon. When $\alpha=1$, we get the HC~ribbon. Observe that $\Phi_\alpha(x)$ is only defined for $x\in [-1,1]$. Therefore, to define the $\Phi$-ribbon, one has to restrict to functions whose range is $[-1,1]$. The following theorem relates $\fR_{\Phi_\alpha}(X_{[k]})$ and $\fR_{\varphi_\alpha}(X_{[k]})$.

\begin{theorem}\label{thm:HC-to-MC}
For all $\alpha\in (1, 2]$ we have $\fR_{\varphi_\alpha}(X_1, \dots, X_k) = \fR_{\Phi_\alpha}(X_1, \dots, X_k)$.
\end{theorem}

A consequence of the above theorem is that by varying the parameter $\alpha$ from $1$ to $2$, the $\Phi$-ribbon varies from the HC~ribbon to the MC~ribbon. 
The proof of this theorem is given in Appendix~\ref{app:HC-to-MC}.

%*********************************************
\section{Strong data processing inequalities}

Let us focus on the bipartite case, namely, $k=2$. Suppose that we have two random variables $(X, Y)$. Then for any function $f_X$ of $X$, we have
$$H_\Phi(f)\geq H_\Phi(\E[f|Y]).$$
This inequality can be thought of as a data processing inequality. Indeed, when $\Phi(t)=1-h((1+t)/2)$, according to Example~\ref{ex:1}, this inequality is equivalent to $I(U; X)\geq I(U;Y)$ assuming the Markov chain $U-X-Y$.
Here, we are interested in how tight this inequality is.  

\begin{definition}
For any convex function $\Phi\in \mathscr F$ define
$\eta_{\Phi} (X, Y)$
to be the smallest (the infimum of) $\lambda\geq 0$ such that for any function $f_X$ (whose range is in the domain of $\Phi$)  we have
$$\lambda H_\Phi(f) \geq H_\Phi(\E[f|Y]).$$
We call $\eta_\Phi(X, Y)$ the $\Phi$-strong data processing inequality constant ($\Phi$-SDPI constant).
\end{definition}

We borrowed the term $\Phi$-SDPI constant from Raginsky~\cite{Raginsky} who defines almost the same invariant. The only difference is that in the definition of~\cite{Raginsky} it is assumed that $\E[f]=1$. This extra assumption, however, does not make a difference in the interesting example of $\Phi(t)=t^\alpha$ as $f_X$ can be scaled (as mentioned in Section \ref{subsection:Phi-examples}).

From the convexity of $\Phi$, it is clear that $\eta_{\Phi}(X, Y)\in [0,1]$. Moreover, if $X$ and $Y$ are independent we have $\eta_\Phi(X, Y)=0$. Also, $\eta_\Phi(X, Y)=1$ if $X=Y$, or more generally if $X$ and $Y$ have explicit common data ($f(X)=g(Y)$ for some non-constant $f$ and $g$). 

When $\Phi(t)=t^2$, $\eta_\Phi(X, Y)$ is nothing but $\rho^2(X, Y)$ as shown in~\cite{Renyi1, Renyi2}. Moreover, for the choice of $\Phi(t) = 1- h((1+t)/2)$, $\eta_\Phi(X, Y)$ is known~\cite{AGKN} to be equal to $s^*(X, Y)$ defined in~\cite{AhlswedeGacs}.

\begin{example}\label{example:DSBS}
The \emph{doubly symmetric binary source} with parameter $\lambda$ denoted by $DSPB(\lambda)$ is defined as follows: $X$ and $Y$ are binary and uniform, and 
$$p(X=0, Y=1)=p(X=1, Y=0)=\frac{1-\lambda}{4}.$$ 
It is known that $\rho^2(X, Y)=s^*(X, Y)=\lambda^2$. We will show in Appendix~\ref{appendix:sphi} that  
$$\eta_\Phi(X, Y) = \lambda^2,$$
holds for all $\Phi\in \mathscr F$.
\end{example}

%Assume that we only have two random variables $X_1$ and $X_2$, \emph{i.e.,} $k=2$.  If we compute the infimum of $(1-\lambda_1)/\lambda_2$ over $(\lambda_1, \lambda_2)$ in the MC~ribbon $\fS(X_1, X_2)$, we get the maximal correlation. If we do the same optimization over the HC~ribbon, we get $s^*(X_1;X_2)$, the so-called strong data processing constant \cite{AhlswedeGacs, AGKN}. We show that the infimum of $(1-\lambda_1)/\lambda_2$ over $(\lambda_1, \lambda_2)$ in the $\Phi$-ribbon gives us $\eta_\Phi(X_1;X_2)$, the $\Phi$-strong data processing constant of Raginsky \cite{Raginsky}.  Since the $\Phi$-ribbon satisfies the data processing and tensorization properties, $\eta_\Phi(X_1;X_2)$ also satisfies these two properties. But only the data processing of $\eta_\Phi(X_1;X_2)$ is new (tensorization has been shown by Raginsky). In Appendix \ref{appendix:sphi} we study properties of  a generalized version of $\eta_\Phi(X_1;X_2)$. 

Let us start investigating properties of $\eta_\Phi(X, Y)$ with two results already proved in~\cite{Raginsky}.
The proof of the following proposition is an immediate consequence of Lemma~\ref{lem:convexity-H-Phi-channel}.

\begin{proposition}[\cite{Raginsky}]\label{prop:convexity-s-channel}
For any convex $\Phi$, and fix $p_X$ the function
$$p_{Y|X}\mapsto \eta_{\Phi} (X, Y),$$
is convex. 
\end{proposition}

It is well-known that $s^*(X, Y)\geq \rho^2(X, Y)$. The following theorem is a generalization of this fact.

\begin{theorem}[\cite{Raginsky}]\label{thm:s-Phi-rho}
For any $\Phi \in \mathscr F$ we have 
$$\eta_\Phi(X, Y) \geq \rho^2(X, Y),$$
where $\rho(X, Y)$ is the maximal correlation.
\end{theorem}

\begin{proof}
Since $\rho^2(X, Y) = \eta_{\Psi}(X, Y)$ for $\Psi(t)=t^2$
we need to show that for any $f_{X}$ with $\E[f]=0$ we have 
\begin{align}\label{eq:eta-rho-0}
\eta_{\Phi}(X, Y) \Var [f] \geq \Var [\E[f|Y].
\end{align}

Take some $c$ in the interior of the domain of $\Phi$, and consider the function $g_X= c+ \epsilon f_X$ for small $|\epsilon|>0$. Then by definition we have 
$$\eta_{\Phi}(X, Y) H_\Phi(g) \geq H_\Phi(\E[g|Y]).$$
Now~\ref{eq:eta-rho-0} follows by applying Lemma~\ref{eq:phi-entropy-taylor} on both sides of this inequality.
\end{proof}

We can now state our main result about the $\Phi$-SDPI constant. 

\begin{theorem}\label{thm:s-Phi-ribbon}
Let $\Phi\in \mathscr F$ be a convex function defined on some \emph{compact} interval. Then we have
\begin{align}\label{eq:s-ribbon}
\eta_{\Phi}(X, Y) = \inf  \frac{1-\lambda_1}{\lambda_2},
\end{align}
where the infimum is taken over all $(\lambda_1, \lambda_2)\in \fR_\Phi(X, Y)$ with $\lambda_2\neq 0$.
\end{theorem}

%In this theorem we have an extra assumption that the function $\Phi$ is Lipschitz continuous. This assumption is satisfied for all functions of interest except $\Phi(x) = 1-h((1+x)/2)$ for which~\eqref{eq:s-ribbon} is already known. 

This theorem for $\Phi(t)=1-h((1+t)/2)$ is derived from the results of~\cite{AhlswedeGacs} and~\cite{AGKN}. For $\Phi(t)=t^2$ it is proved in~\cite{OurPaper}.  

This theorem has a technical assumption, that the domain of $\Phi$ is compact. We do not know whether the theorem holds without this restriction. Nevertheless, this assumption is already satisfied (or can be assumed without loss of generality) for all examples given in Section~\ref{subsection:Phi-examples}. For instance, for $\Phi(t)=t^2$, without loss of generality we can restrict the domain of $\Phi$ to $[-1,1]$ by properly scaling the function $f$. To see this observe that the equation
$
H_\Phi(f)\geq \sum_{i=1}^k\lambda_i H_\Phi(\E[f|X_i])
$
holds for $f$ if and only if it holds for a scaled version of $f$.

\begin{proof}
Let $(\lambda_1, \lambda_2)\in \fR_\Phi(X, Y)$ with $\lambda_2\neq 0$. Then for any function $f_X$ of $X$ we have 
$$H_\Phi(f) \geq \lambda_1H_\Phi(\E[ f|X])  + \lambda_2H_\Phi(\E[ f|Y]).$$
Since $f$ is taken to be a function of $X$ only, we have $f=\E[f|X]$. Therefore,
$$\frac{1-\lambda_1}{\lambda_2} H_\Phi(f)\geq H_\Phi(\E[ f|Y]).$$
Thus we have $(1-\lambda_1)/\lambda_2\geq \eta_\Phi(A, B)$, and then
$$\eta_{\Phi}(X, Y) \leq   \inf  \frac{1-\lambda_1}{\lambda_2}.$$

To prove the inequality in the other direction we show that for any $\delta>0$, we have
$$\eta_{\Phi}(A, B) + \delta \geq   \inf  \frac{1-\lambda_1}{\lambda_2}.$$
To show this, it suffices to argue that there exists $n$ such that the pair $(\lambda_1^{(n)}, \lambda_2^{(n)})$ given by
$$\lambda_1^{(n)} = 1- \frac{\eta_\Phi + \delta}{n}, \qquad \lambda_2^{(n)} = \frac{1}{n},$$
where $\eta_\Phi=\eta_\Phi(X, Y)$,
belongs to $\fR_\Phi(X, Y)$.
Suppose that this is not the case. Then, for any $n$ there is a function $f_{AB}^{(n)}$ such that 
\begin{align*}%\label{eq:counter-example-f-n}
H_\Phi(f^{(n)}) < \lambda_1^{(n)}H_\Phi\big(\E\big[f^{(n)}\big|X\big]\big)  + \lambda_2^{(n)}H_\Phi\big(\E\big[ f^{(n)}\big|Y\big]\big).
\end{align*}
Using the chain rule, this inequality can be rewritten as 
\begin{align*}
\frac{1-\lambda_1^{(n)}}{\lambda_2^{(n)}}
H_\Phi\big(\E\big[f^{(n)}\big|X\big]\big) + \frac{1}{\lambda_2^{(n)}} H_\Phi\big(f^{(n)}|X  \big)  < H_\Phi\big(\E\big[ f^{(n)}\big|Y\big]\big),
\end{align*}
or equivalently as
\begin{align*}
(\eta_\Phi+\delta)
H_\Phi\big(\E\big[f^{(n)}\big|X\big]\big) + n H_\Phi\big(f^{(n)}|X  \big) < H_\Phi\big(\E\big[ f^{(n)}\big|Y\big]\big).
\end{align*}
Since $\Phi$-entropy is non-negative, we infer from this inequality that
\begin{align}\label{eq:s-098-a}
(\eta_\Phi+\delta)
H_\Phi\big(\E\big[f^{(n)}\big|X\big]\big) < H_\Phi\big(\E\big[ f^{(n)}\big|Y\big]\big),
\end{align}
and
\begin{align}\label{eq:s-098-b}
 n H_\Phi\big(f^{(n)}|X  \big) < H_\Phi\big(\E\big[ f^{(n)}\big|Y\big]\big).
\end{align}

Since $\mathcal X$ and $\mathcal Y$ are assumed to be finite, and the images of functions $f^{(n)}$ are in a compact interval (i.e., the domain of $\Phi$), there is an increasing sequence $\{n_{k}: k\geq 1\}$ such that 
$$\lim_{k\rightarrow \infty} f^{(n_k)}(x, y) = \hat f(x, y), \qquad \forall x, y, $$
for some function $\hat f$.

$\Phi$ is continuous and defined on a compact interval. Thus, there is a constant $M>0$, such that $|\Phi(t)|\leq M$ for all $t$. As a result, the $\Phi$-entropy of any function is at most $2M$. Then from~\eqref{eq:s-098-b} we have
\begin{align}\label{eq:f-close-function-A}
H_\Phi\big(f^{(n)}|X  \big)  \leq \frac{2 M}{n}, \qquad \forall n.
\end{align}
Then, by a continuity argument we conclude that $H_\Phi(\hat f|X)=0$, i.e., $\E[\Phi(\hat f)] = \E_X\big[  \Phi(\E[\hat f|X])\big]$. From this equality and the fact that $\Phi\in\mathscr F$ is strictly convex, we infer that $\hat f=\hat f_X$ is a function of $X$ only. 

Next, using~\eqref{eq:s-098-a} we find that
$$(\eta_\Phi + \delta)H_\Phi(\E[\hat f|X]) = (\eta_\Phi + \delta)H_\Phi(\hat f) 
 \leq H_\Phi(\E[\hat f|Y]).
$$
Moreover, since $\hat f$ is a function of $X$, by the definition of $\eta_\Phi(X, Y)$ we have
$$ H_\Phi(\E[\hat f|Y])\leq \eta_\Phi H_\Phi(\hat f).
$$
Putting these two inequalities together we conclude that $H_\Phi(\hat f)=0$, which again by the strict convexity of $\Phi$ imply that $\hat f$ is a constant, i.e.,
$$\lim_{k\rightarrow \infty} f^{(n_k)} = c,$$
for some constant $c$. 

Observe that $\E[f^{(n)}|X]$ is not a constant since otherwise $H_\Phi(\E[f^{(n)}|X]) =0$ and $H_\Phi(f^{(n)})= H_\Phi(f^{(n)}|X)$. In this case, from~\eqref{eq:s-098-b} we find that $nH_\Phi(\E[f^{(n)}|Y])\leq nH_\Phi(f^{(n)}) < H_\Phi(\E[f^{(n)}|Y])$ which is a contradiction.  As a result, $\E[f^{(n)}|X]$ is not a constant and  $\Var_X\E[f^{(n)}|X]> 0$. Then for every $k$ we can write 
$$f^{(n_k)} = c_k+ \epsilon_k g^{(k)},$$
such that 
$$c_k= \E \big[f^{(n_k)}\big],\qquad \E \big[g^{(k)}\big]=0, \qquad \epsilon_k =\sqrt{\Var_X\E\big[f^{(n_k)}\big|X\big]}>0, \qquad \Var_X\E\big[g^{(k)}\big| X\big]=1.$$ 
Observe that $\lim_{k\rightarrow \infty} \epsilon_k=0$ since $f^{(n_k)}$, and then $\E[f^{(n_k)}| X]$ converge to constant functions. We also have 
$\lim_{k\rightarrow \infty}c_k=c$. 
Moreover, $g^{(k)}$'s are uniformly bounded since they have zero expectation and $\Var_X\E\big[g^{(k)}\big| X\big]=1$.

Now using Lemma~\ref{eq:phi-entropy-taylor} we find that 
$$\Big|      H_\Phi(f^{(n_k)}) - \frac{1}{2}\Phi''(c_k)\Var\big[g^{(k)}\big]\epsilon_k^2    
\Big|= O(\epsilon_k^3).$$
Here, in particular, we use the fact that $g^{(k)}$'s are uniformly bounded. 
We similarly have 
\begin{align*}
\Big|H_\Phi\big( \E[f^{(n_k)}|X]  \big) - \frac{1}{2} \Phi''(c_k) \Var_X \big[\E[g^{(k)}|X]\big] \epsilon_k^2\Big| =O(\epsilon_k^3),
\end{align*}
and
\begin{align*}
\Big|H_\Phi\big( \E[f^{(n_k)}|Y]  \big) - \frac{1}{2} \Phi''(c_k) \Var_Y \big[\E[g^{(k)}|Y]\big] \epsilon_k^2\Big| =O(\epsilon_k^3),
\end{align*}
Using these in~\eqref{eq:s-098-a} and~\eqref{eq:f-close-function-A}, and noting that $\epsilon_k>0$ and that $0<\Phi''(c_k)< \Phi''(c)+1$ for sufficiently large $k$, we find that
\begin{align}\label{eq:con-01}
(\eta_\Phi+\delta)\Var_X \big[\E[g^{(k)}|X]\big] < \Var_Y \big[\E[g^{(k)}|Y]\big] +O(\epsilon_k),
\end{align}
and
\begin{align}\label{eq:con-02}
\Var\big[g^{(k)}|X\big]  < \frac{M'}{n_k} + O(\epsilon_k),
\end{align}
for some constant $M'$.

As mentioned above the functions $g^{(k)}$ are uniformly bounded. Then there is an increasing sequence $\{k_j: j\geq 1\}$ such that 
$$\lim_{j\rightarrow \infty}g^{(k_j)} = \hat g,$$
for some function $g_{XY}$. 
Then using~\eqref{eq:con-02} we have $\Var[\hat g|X]=0$, i.e., $\hat g=\hat g_X$ is a function of $X$ only. On the other hand, since $\Var[g^{(k)}|X]=1$, for all $k$, we have $\Var[\hat g|X] = \Var [\hat g]=1$, i.e., $\hat g$ is not a constant.

Next using~\eqref{eq:con-01} we find that
$$(\eta_\Phi+\delta) \Var[\hat g]\leq \Var[\E[\hat g| Y]].$$
Also, since $\rho^(X, Y) = \eta_\Psi(X, Y)$ for $\Psi(t)=t^2$, and $\hat g$ is a function of $X$ we have
$$\Var[\E[\hat g| Y]]\leq \rho^2(X, Y) \Var[\hat g].$$
Comparing the above two inequalities and using $\Var[\hat g]=1$ we conclude that 
$\eta_\Phi(X, Y) + \delta\leq \rho^2(X, Y)$, which is in contradiction with Theorem~\ref{thm:s-Phi-rho}. We are done. 

\end{proof}

We now state the tensorization and monotonicity properties of the $\Phi$-SDPI constant. 

\begin{theorem}\label{thm:tensor-monoton-eta-phi}
For any $\Phi\in \mathscr F$, the $\Phi$-SDPI constant $\eta_{\Phi} (X, Y)$ satisfies the followings:
\begin{enumerate}[label={\rm (\roman*)}]
\item Monotonicity: If $X-A-B-Y$ forms a Markov chain, then $\eta_{\Phi}(X, Y)\leq \eta_{\Phi}(A, B)$.
\item Tensorization: If $p_{ABXY}= p_{AB}\cdot p_{XY}$ then 
$$\eta_{\Phi}(AX, BY) = \max\big\{     \eta_{\Phi}(A, B), \eta_{\Phi}(X, Y) \big\}.$$
\end{enumerate}
\end{theorem}

The tensorization of $\eta_\Phi(X, Y)$ is already proved in~\cite{Raginsky}.
Moreover,  this theorem, in the case that the domain of $\Phi$ is compact, is a simple corollary of Theorem~\ref{thm:s-Phi-ribbon} and the monotonicity and tensorization properties of the $\Phi$-ribbon. 

In Appendix~\ref{appendix:sphi} we give a generalization of the SDPI constant  associated to two functions $\Phi, \Psi$. This constant which we denote by $\eta_{\Phi, \Psi}(X, Y)$, coincides with $\eta_\Phi(X, Y)$ when $\Phi=\Psi$. We prove in Appendix~\ref{appendix:sphi} that $\eta_{\Phi, \Psi}(X, Y)$ satisfies the tensorization and monotonicity properties, from which Theorem~\ref{thm:tensor-monoton-eta-phi} follows as a special case.

\subsection{Example: sums of i.i.d.\ random variables}

Let $X_1, \dots, X_n$ be $n$ i.i.d.\ random variables. For any  $1\leq \ell\leq n$ define 
$$S_\ell = X_1 + \cdots + X_\ell.$$
It is known~\cite{DKS01} that for any $1\leq m\leq n$ we have $\rho^2(S_n, S_m)= \frac{m}{n}$. Moreover, recently it is shown in~\cite{KamathNair15} that $s^*(S_n, S_m) = \frac{m}{n}$. In the following we prove a similar result for all $\Phi\in \mathscr F$.

\begin{theorem}
Let $X_1, \dots, X_n$ be in i.i.d.\ random variables. Then for any $1\leq m\leq n$ and any $\Phi\in \mathscr F$ we have
$$\eta_\Phi(S_n, S_m)=\frac{m}{n}.$$
\end{theorem}

To prove this theorem we borrow ideas from~\cite{KamathNair15}.

\begin{proof}
By Theorem~\ref{thm:s-Phi-rho} we already know that  $\eta_\Phi(S_n, S_m)\geq \rho^2(S_n, S_m)= m/n$. Then it suffices to show that $\eta_\Phi(S_n, S_m)\leq  m/n$. For this we need to verify that for any function $f=f(S_n)$ we have
\begin{align}\label{eq:SDPI-iid}
\frac{m}{n} H_\Phi(f) \geq H_\Phi(\E[f|S_m]).
\end{align}

Observe that for any $1\leq \ell\leq n$, the conditional distribution of $S_n$ given  $X_{[\ell]} = x_{[\ell]}$, for any $(x_1, \dots, x_\ell)$, is identical to the conditional distribution of $S_n$ given $S_\ell =  x_1+\cdots + x_\ell$. Therefore, we have $H_\Phi(f|X_{[\ell]}) = H_\Phi(f| S_\ell)$. Then using the chain rule 
$$H_\Phi(f) = H_\Phi(\E[f|X_{[\ell]}]) + H_\Phi(f|X_{[\ell]}) = H_\Phi(\E[f|S_\ell]) +H_\Phi(f| X_\ell),$$
we find that 
$$H_\Phi(\E[f|X_{[\ell]}]) = H_\Phi(\E[f| S_\ell]).$$
Let us denote the above quantity by $c_\ell$. We claim that 
\begin{align}\label{eq:sdpi-iid-c-ell}
c_{\ell+1} - c_\ell \geq c_\ell - c_{\ell-1}, \qquad0\leq \ell\leq n.
\end{align}
To prove our claim, note that since $X_i$'s are i.i.d.\ we have 
$H_\Phi(\E[f| X_{[\ell]}]) = H_\Phi(\E[f| X_{[\ell-1]}, X_{\ell+1}])$. Therefore, by the chain rule we have
\begin{align*}
c_{\ell+1} - 2c_\ell + c_{\ell-1} &  = H_\Phi(\E[f| X_{[\ell-1]}, X_\ell, X_{\ell+1}]) - H_\Phi(\E[f| X_{[\ell-1]}, X_{\ell+1}])\\
&\quad\, - H_\Phi(\E[f| X_{[\ell-1]}, X_\ell]) + H_\Phi(\E[f| X_{[\ell-1]}])\\
& = H_\Phi\big(\E[f| X_{[\ell-1]}, X_\ell, X_{\ell+1}]  \big| X_{[\ell-1]}, X_{\ell+1}\big)\\
& \quad\, - H_\Phi\big(\E[f| X_{[\ell-1]}, X_\ell] \big| X_{[\ell-1]}\big).
\end{align*}
Let us define $g=\E[f| X_{[\ell-1]}, X_\ell, X_{\ell+1}]$. Then we have
$$c_{\ell+1}- 2c_\ell + c_{\ell-1} = H_\Phi(g|X_{[\ell-1]}, X_{\ell+1}) - H_\Phi
\big(\E[g| X_{[\ell-1]}, X_\ell] | X_{[\ell-1]}\big).$$
Now since $X_{[\ell-1]}, X_\ell$ and $X_{\ell+1}$ are independent, using part (b) of Lemma~\ref{lem:key-lemma-phi-entropy} we arrive at $c_{\ell+1}-2c_\ell + c_{\ell-1}\geq 0$ and then~\eqref{eq:sdpi-iid-c-ell}.

We prove by induction that 
$$\frac{c_\ell}{\ell}\geq \frac{c_{\ell-1}}{\ell-1}.$$ 
The base case $\ell=2$ is immediate from~\eqref{eq:sdpi-iid-c-ell} and that $c_0=0$.  The induction step follows from
\begin{align*}
c_{\ell+1} - c_\ell &\geq c_\ell- c_{\ell-1}
 \geq c_\ell - \frac{\ell-1}{\ell} c_\ell 
= \frac{1}{\ell} c_\ell.
\end{align*}
We conclude that $c_n/n\geq c_m/m$ since $m\leq n$, which is equivalent to~\eqref{eq:SDPI-iid}.

\end{proof}

%****************************************************************************
\section{Maximal correlation ribbon}

In this section we focus on the function $\Phi(x)=x^2$. We note that $\Phi\in \mathscr F$, so the ribbon $\fR_\Phi(X_{[k]})$ satisfies monotonicity and tensorization. The ribbon $\fR_\Phi(X_{[k]})$ for this particular function is introduced in~\cite{OurPaper} as the maximal correlation ribbon (MC~ribbon) and is denoted by $\fS(X_{[k]})$. So we will use this notation here too. Thus $\fS(X_{[k]})$ is the set of $k$-tuples $(\lambda_1, \dots, \lambda_k)\in [0,1]^k$ such that for all functions $f_{X_{[k]}}$ we have
\begin{align}\label{eq:def-mc-variance-0}
\Var[f] \geq \sum_{i=1}^k\lambda_i\Var_{X_i}\big[\E[f|X_i]\big].
\end{align} 
With no loss of generality, in this definition  we may assume that $\E f=0$.

The following theorem states that the MC~ribbon is the largest possible $\Phi$-ribbon.
\begin{theorem}\label{thm:Phi-ribbon-MC-ribbon}
For any $\Phi\in \mathscr F$ we have $\fR_\Phi(X_1, \dots, X_k)\subseteq \fS(X_1, \dots, X_{k})$.
\end{theorem}

The proof of this theorem is based on Lemma~\ref{eq:phi-entropy-taylor} and is similar to that of Theorem~\ref{thm:s-Phi-rho}. So we skip a detailed proof. 

\subsection{Alternative characterizations of the MC~ribbon}

Next, we discuss alternative characterization of the MC~ribbon. We first show that to compute the MC~ribbon, it suffices to restrict to a special class of functions $f(X_{[k]})$.\footnote{ After the MC ribbon was introduced in \cite{OurPaper} by the authors, the second author had an email exchange with Sudeep Kamath who claimed the statement of Proposition \ref{prop:lemmaMCrb}. However, Kamath's proof did not appear rigorous to the second author. Independently, this proposition was found and shown by the first author via a different proof technique.
 }

\begin{proposition}\label{prop:lemmaMCrb} 
In the definition of the MC~ribbon $\fS(X_1, \dots, X_k)$ we may restrict to functions $f_{X_{[k]}}$ that of the form $f(X_{[k]})=f_1(X_1)+\cdots +f_k(X_k)$ where $f_i$ is a function of $X_i$ only.
\end{proposition}

\begin{proof}
Let $\mathcal F_{X_{[k]}}$ be the linear space of all functions over $\mathcal X_1\times \dots \times \mathcal X_k$ equipped with the inner product:
$$\langle f, g\rangle := \E[fg].$$ 
Let $\mathcal F_{X_i}\subseteq \mathcal F_{X_{[k]}}$ be the linear space of all functions that depend only on $X_i$. Observe that for any $i\neq j$, 
$\mathcal F_{X_i}\cap \mathcal F_{X_j}=\text{\rm span}\{ \mathbf{1}_{X_{[k]}}\}$ is the set of constant functions. Moreover, $\mathcal F_{X_i}^0= \mathcal F_{X_i}\cap \mathbf{1}^{\perp}$ consists of all \emph{zero-mean} functions that depend only on $X_i$. Putting these together we have
$$\mathcal F_{X_{[k]}} = \text{\rm span}\{ \mathbf{1}_{X_{[k]}}\} \oplus\left(\bigoplus_{i=1}^k \mathcal F_{X_i}^0 \right)\oplus \mathcal U_{X_{[k]}},$$
where 
$$\mathcal U_{X_{[k]}}=  \bigcap_{i=1}^k \mathcal F_{X_i}^\perp,$$
is the set of functions that are orthogonal to all functions that depend only on one of $X_i$'s. 
Observe that, for any function $u_{X_{[k]}}\in \mathcal U_{X_{[k]}}$ we have $\E[u|X_i]=0$ for all $i\in[k]$,  because $\E(\E[u|X_i]^2)=\E(u\E[u|X_i])$ vanishes since it is the inner product between $u$ and $\E[u|X_i]$, which is a function of $X_i$.

Let $f\in \mathcal F_{X_{[k]}}$ be an arbitrary function with $\E f=0$. By the above decomposition there exist $g_{i}\in \mathcal F_{X_i}^0$ and $u_{X_{[k]}}\in \mathcal U_{X_{[k]}}$ such that 
$$f_{X_{[k]}}=\E[f] +g_{1}+\cdots g_k + u_{X_{[k]}} = g_{1}+\cdots g_k + u_{X_{[k]}}.$$
Let $\tilde{f}=\sum_{i}g_{i}$. Then, $\mathbb{E}[f|X_i]=\mathbb{E}[\tilde{f}|X_i]$ simply because $\E[u|X_i]=0$. Thus,
\begin{align*}
\Var[f]&=\Var[u]+\Var[\tilde{f}],\\
\Var\big[\mathbb{E}[f|X_i]\big]&=\Var\big[\mathbb{E}[\tilde{f}|X_i]\big].
\end{align*}
Now fixing $\tilde f$, since $\Var[u]\geq 0$, the inequality~\eqref{eq:def-mc-variance-0} holds for all $f= \tilde f + u$ if and only if it holds for $\tilde f$.
This completes the proof.

\end{proof}

We can further simplify calculation of the MC~ribbon.

\begin{theorem}\label{thm:mc-equivalent-rep}
The MC~ribbon  $\fS(X_1, \dots, X_k)$ is equal to the set of $k$-tuples $(\lambda_1, \lambda_2, \cdots, \lambda_k)\in[0,1]^k$ such that for all functions $f_1(X_1), \dots, f_k(X_k)$, we have
\begin{align}\Var\left[f_1+\cdots +f_k\right]\leq \sum_{i=1}^{k}\frac{1}{\lambda_i}\Var[f_i].\label{eqn:fc2def}
\end{align}
\end{theorem}

Before giving a proof of this theorem, let us discuss some of its implications.

Firstly, computation of the MC~ribbon using its characterization given in this theorem is much easier because
the dimension of the space of all functions on $\prod_{i=1}^k\mathcal{X}_i$ quickly becomes very large as we increase $k$, the number of variables. However, in the characterization~\eqref{eqn:fc2def} of the MC~ribbon we should search on a space of functions whose dimension scales linearly with $k$. Moreover,  the variance of a conditional expectation in the original definition is replaced by a simple variance that is easier to compute. 

Secondly, Equation~\eqref{eqn:fc2def} can be understood as a \emph{strong Cauchy-Schwarz inequality}. Letting $f_i(X_i)$'s to be arbitrary functions with zero mean, by the Cauchy-Schwarz inequality we have
\begin{align}
\E[\sum_{i=1}^k f_i]^2\leq k\sum_{i=1}^{k}\E[f_i]^2.
\label{eqn:caschwa}
\end{align}
Then by Theorem~\ref{thm:mc-equivalent-rep},  the $k$-tuple $(\frac 1k, \frac 1k, \cdots, \frac 1k)$ belongs to $\fS(X_{[k]})$. Then the MC~ribbon characterizes the extent to which this inequality can be strengthened.

Thirdly, from the definition of the MC~ribbon it is clear that $\fS (X_{[k]})$ is convex. Theorem~\ref{thm:mc-equivalent-rep} says that the set of $k$-tuples $(\lambda_1^{-1}, \dots, \lambda_k^{-1})$ such that $\lambda_{[k]}\in\fS(X_{[k]})$, is convex too. This is a fact that is not clear from the definition of the MC~ribbon.

%\begin{remark} Assuming that $g_i(x_i)=\lambda_if(x_i)$, we can write equation \eqref{eqn:fc2def} as follows:
%\begin{align}\Var[\sum_{i=1}^k \lambda_i g_i]\leq \sum_{i=1}^{k}\lambda_i\Var[g_i].\label{eqn:fc2defn}
%\end{align}
%\end{remark}

\begin{proof}
%[Proof of Theorem \ref{thm:mc-equivalent-rep}] 
Our proof uses Proposition~\ref{prop:lemmaMCrb}. Consider the space of functions $\mathcal F_{X_i}^0$ of all zero-mean functions of $X_i$ (as defined in the proof of Proposition~\ref{prop:lemmaMCrb}). With no loss of generality, we can assume that $p_{X_i}(x_i)>0$ for all $x_i$. Hence,  $\mathcal F_{X_i}^0$ is a linear vector space with dimension $|\mathcal X_i|-1$. Let $\{f_{ij}: j=1,\dots, |\mathcal X_i|-1\}$ be an orthonormal basis for  $\mathcal F_{X_i}^0$.
Thus, $f_{ij}$ is a function of $X_i$ with zero mean and unit variance, and we have $\E[f_{ij}f_{ij'}]=0$ for $j\neq j'$.
Define
$$m_{i_1j_1;i_2j_2}=m_{i_2j_2;i_1j_1}=\mathbb{E}\big[f_{i_1j_1}(X_{i_1})f_{i_2j_2}(X_{i_2})\big].$$

Any arbitrary zero-mean function $g_i(X_i)$ can be expressed as
$$g_i=\sum_{j=1}^{|\mathcal X_i|-1}c_{ij}f_{ij},$$
for some real coefficients $c_{ij}$. Then we have
$\Var[g_i]=\sum_{j} c_{ij}^2,$
and
\begin{align*}
\E\big[g_{i_1}g_{i_2}\big] =\sum_{j_1, j_2} c_{i_1j_1}c_{i_2j_2}m_{i_1j_1;i_2j_2}.
\end{align*}
Moreover, for any function $u$ with zero mean, $\Var[\E[u|X_i]]$, i.e., the squared length of $\E[u|X_i]$, is given by $\sum_{j}\big(\mathbb{E}[u f_{ij}]\big)^2.$
Hence, 
\begin{align}
\Var\bigg[\mathbb{E}\bigg[\sum_{i_2=1}^kg_{i_2}\Big|X_{i_1}\bigg]\bigg]&=\sum_{j_1}\bigg(\mathbb{E}\bigg[\sum_{i_2=1}^kg_{i_2} f_{i_1j_1}\bigg]\bigg)^2\\
&=\sum_{j_1}\bigg(\mathbb{E}\bigg[\sum_{i_2}\sum_{j_2}c_{i_2j_2}f_{i_2j_2} f_{i_1j_1}\bigg]\bigg)^2\\
&=\sum_{j_1}\bigg(\sum_{i_2}\sum_{j_2}c_{i_2j_2}m_{i_1j_1;i_2j_2}\bigg)^2.
\end{align}
Putting all these together and using Proposition~\ref{prop:lemmaMCrb}
we find that the MC~ribbon is the set of $k$-tuples $\lambda_{[k]}$ such that for all $c_{ij}$'s we have
\begin{align}
\sum_{i_1,i_2,j_1,j_2} c_{i_1j_1}c_{i_2j_2}m_{i_1j_1;i_2j_2}\geq \sum_{i_1}\lambda_{i_1}\sum_{j_1}\bigg(\sum_{i_2}\sum_{j_2}c_{i_2j_2}m_{i_1j_1;i_2j_2}\bigg)^2.\label{eqn:eqn:eqn1}
\end{align}
Similarly, the ribbon defined by~\eqref{eqn:fc2def} is the set of $k$-tuples $\lambda_{[k]}$ such that for all $c_{ij}$'s we have
\begin{align}
\sum_{i_1,i_2,j_1,j_2} c_{i_1j_1}c_{i_2j_2}m_{i_1j_1;i_2j_2}\leq \sum_{i}\frac{1}{\lambda_{i}}\sum_{j}c_{ij}^2.
\label{eqn:eqn:eqn2}
\end{align}
Therefore, it remains to show that the ribbons defined by equations \eqref{eqn:eqn:eqn1} and \eqref{eqn:eqn:eqn2} are the same.

Let $M$ be the matrix whose $(i_1j_1, i_2j_2)$ entry is $m_{i_1j_1;i_2j_2}$. Note that $M$ is positive definite since it is a Gram matrix. Also let $\Lambda$ be the diagonal matrix whose $(ij, ij)$ entry equals  $\lambda_{i}$.
Then, $\lambda_{[k]}$ satisfies~\eqref{eqn:eqn:eqn2} for all $c_{ij}$'s iff $M\leq \Lambda^{-1}$, i.e., iff $(\Lambda^{-1}-M)$ is positive semidefinite. By the operator monotonicity of the function $t\mapsto -t^{-1}$, this is equivalent with $M^{-1}\geq \Lambda$ as well as $M\geq M\Lambda M$. Now a straightforward calculation shows that $M\geq M\Lambda M$ is equivalent with~\eqref{eqn:eqn:eqn1} for all $c_{ij}$'s. This completes the proof. 
\end{proof}

Let us go back to the characterization of MC~ribbon given in  Proposition~\ref{prop:lemmaMCrb}. Let $f= f_1(X_1) +\cdots +f_k(X_k)$ be such that $f_i(X_i)$ has zero mean for $i=1, \dots, k$.  Observe that $\Var[\mathbb{E}[f|X_i]]$ is the squared length of the projection of $f$ onto the space of all zero-mean functions of $X_i$, namely $\mathcal F_{X_i}^0$. Then 
$\Var[\mathbb{E}[f|X_i]]$ can be bounded from below by the squared of the inner product of $f$ with some unit vector in $\mathcal F_{X_i}^0$. 
 In particular, for the choice of  
\begin{align}\label{eq:hat-f-normalized}
\hat{f}_i=\frac{f_i}{\sqrt{\Var[f_i]}},
\end{align}
as a unit vector in $\mathcal F_{X_i}^0$ we have
$$\Var[\mathbb{E}\big[f|X_i]\big]\geq 
\mathbb{E}\big[f \hat{f}_i\big]^2.$$
The following theorem shows that using this lower bound in Proposition~\ref{prop:lemmaMCrb} gives another equivalent representation of the MC~ribbon.

\begin{theorem} \label{thm:2-rep-MC-ribbon}
Let $\fS'(X_1, \dots, X_k)$ be the set of the $k$-tuples $(\lambda_1, \dots, \lambda_k)\in[0,1]^k$ such that for all functions $f_i(X_i)$, $i=1, \dots, k$, we have
$$\Var[f]\geq \sum_{i=1}^{k}\lambda_i\mathbb{E}\big[f\hat{f}_i\big]^2,$$
where $f=\sum_{i=1}^k f_i$ and $\hat f_i$ is given by~\eqref{eq:hat-f-normalized}. Then we have $\fS'(X_{[k]})=\fS(X_{[k]})$.
\end{theorem}

The proof of this theorem is based on similar ideas as in the proof of Theorem~\ref{thm:mc-equivalent-rep}, so we leave it for Appendix~\ref{app:proof-thm:2-rep-MC-ribbon}.

\subsection{Extreme MC~ribbons}

It is well-known that $\rho(X, Y)=0$ if and only if $X$ and $Y$ are independent. The following proposition is a generalization of this fact. 

\begin{proposition}\label{prop:MC-extreme1}
$\fS(X_{1}, \dots, X_k)$ is equal to $[0,1]^k$ if and only if $X_1, X_2, \dots, X_k$ are \emph{pairwise} independent.
\end{proposition}

Note that the HC~ribbon is equal to $[0,1]^k$ if $X_i$'s are \emph{mutually} independent. Thus, MC~ribbon and HC~ribbon behave completely differently and provide different characterizations of the correlations in $(X_1, \dots, X_k)$ when $k\geq 3$.
%Explicit calculation of the MC~ribbon for $k=3$ in Theorem \ref{BinaryBinaryTernary} shows that to compute $\fS(X_{[k]})$ for a given $p(x_{[k]})$, it is not enough to know only the pairwise marginals $p(x_i, x_j)$ for $i,j\in [k]$. However, when the pairwise marginals become independent $p(x_i, x_j)=p(x_i)p(x_j)$ for any $i\neq j$, then $\fS(X_{[k]})$ becomes equal to $[0,1]^k$. 

The tensorization property of the MC~ribbon implies that $$\fS(X_1, X_2, \dots, X_k)=\fS(M_1X_1, M_2X_2, \dots, M_kX_k)$$
if $M_i$s are independent of $X_i$s, and $M_i$s are pairwise independent. 
 This fact shows that we can prove infeasibility for the non-interactive distribution simulation problem \cite{KamathAnantharam}, in the presence of ``private" randomnesses $M_j$ that are pairwise (and not mutually independent).

\begin{proof}
We may use Theorem~\ref{thm:mc-equivalent-rep}. If $X_i$'s are pairwise independent we clearly have $\Var[f_1+\cdots + f_k]= \Var[f_1]+\cdots + \Var[f_k]$ and then $\fS(X_{[k]}) = [0,1]^k$. Conversely, suppose that $\Var[f_1+\cdots + f_k]\leq \Var[f_1]+\cdots + \Var[f_k]$, for all $f_i(X_i)$'s. By letting $f_\ell$'s to be equal to the zero function except the $i$-th and $j$-th ones, we find that $\Var[f_i+f_j]\leq \Var[f_i]+\Var[f_j]$ for all $f_i(X_i)$ and $f_j(X_j)$. This means that $X_i$ and $X_j$ are independent. 
\end{proof}

Recall that $\fS(X_1, \dots, X_k)$ always contains $\{\lambda_{[k]}|\, \lambda_i\geq 0,\, \sum_{i}\lambda_i\leq 1\}$. The following proposition characterizes the other extreme for the MC~ribbon.

\begin{proposition}
$\fS(X_{[k]})=\{\lambda_{[k]}|\, \lambda_i\geq 0,\, \sum_{i}\lambda_i\leq 1\}$  if and only if $X_i$'s have a \emph{common part}, i.e., there are \emph{non-constant} functions $g_i(X_i)$, $i=1, \dots, k$, such that $g_1(X_1)=\cdots =g_k(X_k)$ with probability one. \label{prop:MC-common-part}
\end{proposition}

\begin{proof}
If $X_i$'s have a common part, for $f(X_{[k]})=g_1(X_1)=\cdots =g_k(X_k)$ we have 
$\Var[f] = \Var[\E[f|X_i]]$. Therefore, by the definition of the MC~ribbon we have $\fS(X_{[k]})=\{\lambda_{[k]}|\, \lambda_i\geq 0,\, \sum_{i}\lambda_i\leq 1\}$.

Conversely, assume that $\fS(X_{[k]})=\{\lambda_{[k]}: \lambda_i\geq 0, \sum_{i}\lambda_i\leq 1\}$. Then, the maximum value of $\lambda$ such that $(\lambda, \lambda, \dots, \lambda)$ is in $\fS(X_{[k]})$, is equal to $1/k$. This maximum value of $\lambda$ can be written as
$$\lambda_{\max}=\inf_{f} \frac{
\Var[f]}{\sum_{i=1}^{k}\Var\big[\mathbb{E}[f|X_i]\big]}.$$
Observe that by scaling $f$, we can restrict the infimum to functions satisfying $\Var[f]=1$. Then by a compactness argument, the infimum is achieved at some (non-constant) function $f$. For this function we have 
$\sum_{i=1}^{k}\Var\big[\mathbb{E}[f|X_i]\big]=k \Var[f].
$
This means that all the inequalities $\Var\big[\mathbb{E}[f|X_i]\big]\leq\Var[f]$ are equality for $f$. 
That is, by the law of total variance, we have
$\Var[f|X_i]=0$. In other words, $f$ is a function of $X_i$ for all $i\in[k]$, and a common part.

\end{proof}

%**************************************************************
\subsection{Examples}\label{sec:calculationMCribbon}

We now compute the MC~ribbon for some examples of $(X_1, \dots, X_k)$. We first focus on the bipartite case.

\begin{proposition}\label{prop:MC-bipartite}
For $k=2$ we have
\begin{align}\label{eq:MC-bipartite}
\fS(X_1, X_2)=\bigg\{(\lambda_1, \lambda_2)\in[0,1]^2 \Big|\, \Big(1-\frac{1}{\lambda_1}\Big)\Big(1-\frac{1}{\lambda_2}\Big)\geq \rho(X_1,X_2)^2\bigg\}.
\end{align}
\end{proposition}

It was proved in~\cite{OurPaper} that the right hand side in~\eqref{eq:MC-bipartite} always contains $\fS(X_1, X_2)$. Here we prove that indeed equality holds. 

\begin{proof}
Using Theorem~\ref{thm:mc-equivalent-rep}, $\fS(X_1, X_2)$ is equal to the set of pairs $(\lambda_1, \lambda_2)$ such that 
$$\Var[f_1 + f_2]\leq \frac{1}{\lambda_1} \Var[f_1] + \frac{1}{\lambda_2}\Var[f_2],$$
for all zero-mean functions $f_1(X_1)$ and $f_2(X_2)$. This inequality is equivalent to 
$$2\E[f_1f_2]\leq \Big(1-\frac{1}{\lambda_1}\Big)\Var[f_1] + \Big(1-\frac{1}{\lambda_2}\Big)\Var[f_2].$$
Then the desired result follows once we note that 
$$\rho(X, Y)= \max \frac{\E[g_1g_2]}{\sqrt{\Var[g_1] \, \Var[g_2]}},$$
where the maximum is over all zero-mean functions $g_1(X_1)$ and $g_2(X_2)$.
\end{proof}

Let us now consider computing the MC~ribbon for multivariate distributions, i.e., for $k\geq 3$. Observe that if $X_i$ is binary (taking values in a binary set), then there is a unique (up to a constant) function $f_i(X_i)$ that has zero mean. Then computing the MC~ribbon using Theorem~\ref{thm:mc-equivalent-rep} is not hard if $X_i$'s are all binary. See Theorem~\ref{thm:Gaussian} for details.

 \paragraph{Binary-Binary-Ternary:} 
Assume that $k=3$, and that $X_1,X_2$ are binary, and $X_3$ is ternary. 
Let $\rho_{ij}=\rho(X_i, X_j)$ be the maximal correlation coefficient between $X_i$ and $X_j$ for distinct $i,j\in\{1,2,3\}$. Since $X_1$ and $X_2$ are binary, there are unique (up to a sign) zero-mean functions $g_1(X_1)$ and $g_2(X_2)$ with unit variance. We choose the sign of such functions $g_1(X_1)$ and $g_2(X_2)$ such that 
\begin{align}\label{eq:g1g2-E}
\mathbb{E}[g_1g_2]=\rho_{12}.
\end{align}
Again by the uniqueness of $g_i$, $i=1,2$, and that $\eta_{\Psi}(X_i, X_3)= \rho^2(X_i, X_3)$, for $\Psi(t)=t^2$, we have
\begin{align}
\rho_{13}^2=\Var\big[\mathbb{E}[g_1|X_3]\big],
\end{align}
and
\begin{align}
\rho_{23}^2=\Var\big[\mathbb{E}[g_2|X_3]\big].
\end{align}
Finally define
\begin{align}\label{eq:def-r123}
r_{12\rightarrow 3}=\mathbb{E}\Big[\mathbb{E}\big[g_1|X_3\big]\,\mathbb{E}\big[g_2|X_3\big]\Big].
\end{align}
Assume that the distribution of $(X_1, X_2, X_3)$ is generic, so that the functions 
$\mathbb{E}[g_1|X_3]$ and $\mathbb{E}[g_2|X_3]$ are linearly independent.

\begin{proposition} \label{prop:binary-binary-ternary} 
Under the assumptions given above, $\fS(X_1, X_2, X_3)$ is the set of triples $(\lambda_1, \lambda_2, \lambda_3)\in [0,1]^3$ such that the followings hold:
$$\Big(\frac{1}{\lambda_1}-1\Big)\Big(\frac{1}{\lambda_3}-1\Big)\geq\rho_{13}^2, $$
$$\Big(\frac{1}{\lambda_2}-1\Big)\Big(\frac{1}{\lambda_3}-1\Big) \geq\rho_{23}^2,$$
$$\left(
\Big(\frac{1}{\lambda_1}-1\Big)\Big(\frac{1}{\lambda_3}-1\Big)-\rho_{13}^2
\right)\left(
\Big(\frac{1}{\lambda_2}-1\Big)\Big(\frac{1}{\lambda_3}-1\Big)-\rho_{23}^2\right)\geq \left(\Big(\frac{1}{\lambda_3}-1\Big)\rho_{12} +r_{12\rightarrow 3}  \right)^2.$$

\end{proposition}

\bigskip

Observe that by the above theorem, the MC~ribbon of $(X_1, X_2, X_3)$ cannot be computed solely based on the marginal distributions of pairwise random variables (compare this with Proposition \ref{prop:MC-extreme1}). 

The proof of this proposition is given in Appendix~\ref{app:binary-binary-ternary}.

%%%%%
\paragraph{Normal distributions:}
Throughout the paper, we considered only discrete random variables, i.e., $\mathcal X_i$'s are finite sets. Nevertheless, the definition of $\Phi$-ribbon can easily be generalized to the continuous case. Moreover, most of the properties that we proved here, except those that are based on compactness, are generalized for continuous random variables as well. Here, our goal is to compute the MC~ribbon for multivariate normal distributions using Theorem~\ref{thm:mc-equivalent-rep}. Since we have not presented the proof of this theorem in the continuous case, the reader may consider the statement of Theorem~\ref{thm:mc-equivalent-rep} as the definition of MC~ribbon for normal distributions. 

%Thus here our goal is to find the set of $k$-tuples
%\begin{align}\Var[\sum_{i=1}^k f_i]\leq \sum_{i=1}^{k}\frac{1}{\lambda_i}\Var[f_i].\end{align}
 %for all functions $f_i:\mathcal{X}_i\mapsto \mathbb{R}$ with finite variance.

Let $(X_1, \dots, X_k)$ be \emph{real} random variables that are either binary (i.e., the alphabet set of $X_i$ is of size two), or normal (i.e., $X_i$'s form a multivariate normal distribution). Let $R$ be the \emph{covariance matrix} of $X_i$'s. That is, $R$ is a matrix whose $(i,j)$-th entry is the Pearson correlation coefficient between $X_i$ and $X_j$:  
$$R_{ij}=\frac{\Cov(X_i, X_j)}{\sqrt{\Var[X_i]\Var[X_j]}},$$
where $\Cov(X_i, X_j) = \E\big[(X_i - \E[X_i]) (X_j - \E[X_j])\big]$. 
Next, given a  $k$-tuple $(\lambda_1, \lambda_2, \cdots, \lambda_k)\in\fS$, we associate to it a diagonal matrix
$\Lambda$ whose $i$-th entry on the diagonal is equal to $\lambda_i$. 

\begin{theorem}\label{thm:Gaussian}
Suppose that $(X_1, \dots, X_k)$ either form a multivariate normal distribution, or are all binary taking values in an alphabet set of size two. Let $R$ be the covariance matrix of $(X_1, \dots, X_k)$ as defined above. Then,   $(\lambda_1,  \dots, \lambda_k)$ belongs to $\fS(X_1, \dots, X_k)$ if and only if
$R\leq \Lambda^{-1}.$
\end{theorem}

The proof of this theorem is based on ideas from~\cite{Lancaster57} and is given in Appendix~\ref{app:Gaussian}.

%****************************************************
\subsection{Another multipartite correlation region} 

Motivated by the form of characterization of $\fS(X_1, \dots, X_k)$ given by Theorem~\ref{thm:mc-equivalent-rep}, we define 
another region associate to a multivariate distribution.

\begin{definition}\label{myRepeatedTheorem2}
For any $(X_1, \dots, X_k)$ we define $\widetilde\fS(X_1, \dots, X_k)$ to be the set of $k$-tuples $(\lambda_1,  \dots, \lambda_k)\in[0,1]^k$ such that for all functions $f_i(X_i)$, $i=1, \dots k$, we have
$$\Var\big[f_1+\cdots + f_k\big]\geq \sum_{i=1}^{k}\lambda_i\Var[f_i].$$
\end{definition}

$\widetilde\fS(X_1, \dots, X_k)$ and the MC~ribbon share the properties of monotonicity and tensorization, yet as we will argue later, they are not identical.

\begin{theorem} \label{thm:tensorization2}
$\widetilde{\fS}(X_{1}, \dots, X_k)$ satisfies the monotonicity and tensorization properties the same as $\fS(X_1, \dots, X_k)$.  
\end{theorem}

This theorem is proved in Appendix~\ref{app:tilde-S-proof}.

\begin{proposition}\label{prop:tilde-S-extreme}
$\widetilde \fS(X_{1}, \dots, X_k)$ is equal to $[0,1]^k$ if and only if $X_1, X_2, \dots, X_k$ are \emph{pairwise} independent.\end{proposition}

The proof of the following proposition is similar to that of Proposition~\ref{prop:MC-extreme1}, so we do not repeat it here. 

The above proposition characterizes one extreme of $\widetilde \fS(X_1, \dots, X_k)$ that is common with the MC~ribbon. To characterize the other extreme of $\widetilde \fS(X_1, \dots, X_k)$,  recall that the MC~ribbon contains all
$\lambda_{[k]}\in [0,1]^k$ with $\sum_{i}\lambda_i\leq 1$. Nevertheless, these points  may not belong to $\widetilde \fS(X_1, \dots, X_k)$. 
Indeed, we can even have $\widetilde \fS(X_{[k]})=\{(0,0,\dots, 0)\}$. To see this, let $k=2$ and assume that $X_1=X_2$ with probability one. Let $f_1=-f_2$ be a non-constant function of both $X_1=X_2$. Then, $f_1+f_2=0$, and thus $\Var[f_1+f_2]=0$. But $\Var[f_1]$ and $\Var[f_2]$ are both positive, so $\widetilde \fS(X_1, X_2)$ contains only $(0,0)$.

\begin{theorem}\label{thm:tilde-S-extreme-0}
$\widetilde{\fS}(X_{[k]})=\{(0,\dots, 0)\}$ if and only if there are zero-mean functions $f_1(X_1), \dots, f_k(X_k)$ that are not all zero and $f_1+\dots + f_k=0$.
\end{theorem}

The proof of this theorem is given in Appendix~\ref{app:tilde-S-extreme-0}.

Based on Theorem~\ref{thm:tilde-S-extreme-0}, for $k=2$, $\widetilde\fS(X_1, X_2)=\{(0,0)\}$ if and only if $X_1$ and $X_2$ have common data. However, for $k\geq 3$ one can find examples of $(X_1, \dots, X_k)$ that do not have common part, yet we have $\widetilde{\fS}(X_{[k]})=\{(0,\dots, 0)\}$.

\begin{example} 
Let $\mathcal X=\mathcal Y=\mathcal Z=\{0,1\}$. Let $0<a,b<1$ be such that $c=a+b<1$. Define $p(x, y, z)$ by
$$p(000)=a,\quad p(110)=b, \quad p(101)=1- c,$$
and $p(xyz)=0$ for $(x,y,z)\notin\{(0,0,0), (1,1,0), (1,0,1)\}$.
%Then the bipartite marginal distributions are
%\begin{center}
%\includegraphics[width=3.2in]{XYZdistribution.pdf}
%\end{center}
%and we have $p(X=0)=a$, $p(Y=0)=\bar b$ and $p(Z=0)=c$.
Observe that $X, Y, Z$ do not have common data.
Now define $f_X, g_Y, h_Z$ as follows: 
\begin{align*}
 f_X(0)=1- a,&\qquad f_X(1)=a,\\
 g_Y(0)=-b,&\qquad g_Y(1)=1-{b},\\
 h_Z(0)={c}-1,& \qquad h_Z(1)=c.
\end{align*}
Then, we have $\E[f_X]=\E[g_Y]=\E[h_Z]=0$, $\E[f_X^2]>0, \E[g_Y^2]>0, \E[h_Z^2]>0$ and $f(X)+g(Y)+h(Z)=0$.  Therefore, $\widetilde{\fS}(X,Y,Z)=\{(0,0,0)\}$. 
\end{example}

In the following theorem we use the same notation as we used in Theorem~\ref{thm:Gaussian}. The proof of this theorem is also similar to that of Theorem~\ref{thm:Gaussian}, so we do not repeat it here. 

\begin{theorem}\label{thm:Gaussian-tilde}
Suppose that $(X_1, \dots, X_k)$ either form a multivariate normal distribution, or are all binary $|\mathcal{X}_i|=2$. Let $R$ be the covariance matrix of $(X_1, \dots, X_k)$ as defined above. Then,   $(\lambda_1,  \dots, \lambda_k)$ belongs to $\fS(X_1, \dots, X_k)$ if and only if
$R\leq \Lambda.$
\end{theorem}

\section{Summary of the results}
In this paper, we defined $\Phi$-ribbon, that generalizes both the  MC and the HC ribbons. We showed that the $\Phi$-ribbon satisfies the monotonicity and tensorization properties. Therefore, $\Phi$-entropy can be utilized in any of the known applications of the HC ribbon and the maximal correlation in network information theory, namely the non-interactive distribution simulation problem \cite{KamathAnantharam} or transmission of correlated sources over a noisy network (such as a MAC channel) \cite{KangUlukus}.  It was shown that the $\Phi$-ribbon relates to Raginsky's $\Phi$-strong data processing constant. Next, we showed that the MC ribbon is the maximal $\Phi$-ribbon, \emph{i.e.,} all $\Phi$-ribbons are subsets of the MC ribbon. This fact motivated further study of the properties of the MC ribbon, and its efficient calculation. In particular, an equivalent characterization of the MC ribbon was given. Inspired by the form of this characterization, we defined another multivariate correlation region that is characterized by maximal correlation when $k=2$, and satisfies the data processing and tensorization properties.

%%%%%%%%%%%%

%******************************************************************

%******************************************************************
%******************************************************************

\bibliographystyle{IEEEtran}

\appendix

%******************************************************************
\section{SDPI constant}\label{appendix:sphi}

In this appendix we prove Theorem~\ref{thm:tensor-monoton-eta-phi} as well as the claim we made in Example~\ref{example:DSBS}. Let us start with the latter.

\begin{proposition}
If $(X, Y)$ are distributed according to $\DSBS(\lambda)$, then for any $\Phi\in \mathscr F$ we have 
$$\eta_\Phi(X, Y) = \lambda^2.$$
\end{proposition}

\begin{proof}
By Theorem~\ref{thm:s-Phi-rho} we know that $\eta_\Phi(X, Y)\geq \rho^2(X, Y)=\lambda^2$. To prove the inequality in the other direction we need to show that for any function $f_X$ we have
$$\lambda^2H_\Phi(f)\geq H_\Phi(\E[f|Y]).$$
Let $m=\E f$, and and define $z$ by $f(0) = m+z$. Then $f(1)= m-z$, and the above inequality reduced to
\begin{align*}
\lambda^2\big(\Phi(m+z) + \Phi(m-z) - 2\Phi(m)\big) \geq \Phi(m+\rho z) + \Phi(m-\rho z) - 2\Phi(m). 
\end{align*}

Let us for $t\geq 0$ define 
$$\psi(t) = \Phi(m+ \sqrt t) + \Phi(m-\sqrt t) - 2\Phi(m).$$
Then the above inequality is equivalent to 
$$\lambda^2\psi(t^2) \geq \psi(\lambda^2 t^2).$$
Since $\psi(0)=0$, this inequality is proven once we show that $\psi$ is convex.  We compute
$$\psi'(t) = \frac{1}{2\sqrt t}\Phi'(m+\sqrt t) - \frac{1}{2\sqrt t} \Phi'(m-\sqrt t),$$
and 
$$\psi''(t) = -\frac{1}{4t^{3/2}} \Big(\Phi'(m+\sqrt t) - \Phi'(m-\sqrt t)\Big)  + \frac{1}{4t}\Big( \Phi''(m+\sqrt t) + \Phi''(m-\sqrt t)  \Big).$$
Then the convexity of $\psi(x)$ is equivalent to  
$$\Phi''(m+\sqrt t) + \Phi''(m-\sqrt t)\geq \frac{1}{\sqrt t} \Big(\Phi'(m+\sqrt t) - \Phi'(m-\sqrt t)\Big).$$
Equivalently we need to show that for any $s\geq 0$ we have
$$\Phi''(m+s) + \Phi''(m-s)\geq \frac{1}{s} \Big(\Phi'(m+s) - \Phi'(m-s)\Big).$$

Define $\xi(s) = \Phi'(m+s) - \Phi'(m-s)$. Then the above inequality can be rewritten as 
$$\xi'(s)\geq \frac{1}{s}\Big( \xi(s) - \xi(0) \Big),$$
which holds if $\xi'(s)$ is increasing since $s\geq 0$. Equivalently we need to prove $\xi''(s)\geq 0$ for all $s\geq 0$. That is, we want 
$$\Phi'''(m+s)\geq \Phi'''(m-s),$$
which holds if $\Phi'''$ is increasing. 

By part (vi) of the definition of class of functions $\mathscr F$, we have $\Phi'''' \Phi''\geq 2\Phi'''^2$. On the other hand, $\Phi$ is convex which means that $\Phi''\geq 0$. Then by this inequality we have $\Phi''''\geq 0$. As a result, $\Phi'''$ is increasing. 
We are done.
 
\end{proof}

Before proving Theorem~\ref{thm:tensor-monoton-eta-phi} let us first generalize the definition of the $\Phi$-SDPI constant. 

\begin{definition}
For any pair of convex functions $\Phi, \Psi$ we define
$\eta_{\Phi, \Psi} (X, Y),$
to be the smallest (the infimum of) $\lambda\geq 0$ such that for any function $f_X$ we have
$$\lambda H_\Phi(f)\geq H_\Psi(\E[f|Y]).$$
\end{definition}

Observe that $\eta_{\Phi, \Psi}(X, Y)$ may be greater than one for arbitrary $\Phi$ and $\Psi$. Moreover,  $\eta_{\Phi, \Psi}(X, Y)$ coincides with $\eta_\Phi(X, Y)$ when $\Psi=\Phi$.

The first property that $\eta_{\Phi, \Psi}$ share with the $\Phi$-SDPI constant is Proposition~\ref{prop:convexity-s-channel}, that for a fixed $p_X$ the function
$$p_{Y|X}\mapsto \eta_{\Phi, \Psi} (X, Y),$$
is convex.  The proof of this fact is again based on Lemma~\ref{lem:convexity-H-Phi-channel}.

\begin{theorem}
Let $\Phi$ and $\Psi$ be convex functions, and assume that at least one of them belongs to~$\mathscr F$. Then the followings hold.
\begin{enumerate}[label={\rm (\roman*)}]
\item Monotonicity: If $X-A-B-Y$ then $\eta_{\Phi, \Psi}(X, Y)\leq \eta_{\Phi, \Psi}(A, B)$.
\item Tensorization: If $p_{ABXY}= p_{AB}\cdot p_{XY}$ then 
$$\eta_{\Phi, \Psi}(AX, BY) = \max\big\{     \eta_{\Phi, \Psi}(A, B), \eta_{\Phi, \Psi}(X, Y) \big\}.$$
\end{enumerate}
\end{theorem}

The proof of is theorem is very similar to that of Theorem~\ref{data-process-tens-Phi}.

\begin{proof}
(i) Let $f_X$ be a function of $X$.  Then $\E[f|A]$ is a function of $A$, and by the definition of $\eta_{\Phi, \Psi}(A, B)$ we have 
\begin{align*}
\eta_{\Phi, \Psi}(A, B) H_\Phi(f)\geq \eta_{\Phi, \Psi}(A, B) H_\Phi(\E[f|A ]) & \geq H_\Psi\big(\E\big[\E[ f|A]\big|B\big]\big) 
 =   H_\Psi(\E[ f|B]),
 \end{align*}
 where the equality follows since we have the Markov chain $X-A-B$. 
Letting $g_B=\E[f|B]$ we have
$$H_\Psi(g) \geq H_\Psi(\E[g|Y]) = H_\Psi(\E[f|Y]),$$
since $X-B-Y$ forms a Markov chain.  
Putting these together we arrive at 
$\eta_{\Phi, \Psi}(A, B) H_\Phi(f) \geq H_\Psi(\E[f|Y])$. As a result, $\eta_{\Phi, \Psi}(X, Y)\leq \eta_{\Phi, \Psi}(A, B)$.

\vspace{.2in}
\noindent
(ii) By restricting to functions that depend only on $A$ or on $X$, it is easy to see that 
$$\eta_{\Phi, \Psi}(AX, BY) \geq  \max\big\{     \eta_{\Phi, \Psi}(A, B), \eta_{\Phi, \Psi}(X, Y) \big\}.$$

Let 
$\lambda=\max\big\{     \eta_{\Phi, \Psi}(A, B), \eta_{\Phi, \Psi}(X, Y) \big\}.$
Then we need to show that $\lambda \geq \eta_{\Phi, \Psi}(AX, BY)$. To prove this we need show that for any function $f_{AX}$ we have
\begin{align}\label{eq:lambda-s-phi-psi}
\lambda H_\Phi(f) \geq H_\Psi(\E[f|BY]).
\end{align}

First assume that $\Psi\in \mathscr F$. Since $\lambda\geq \eta_{\Phi, \Psi}(X, Y)$ and the distribution of $(X, Y)$ does not change when we condition on $A$, we have
\begin{align*}
\lambda H_\Phi(f|A) \geq H_\Psi(\E[f|AY]|A).
\end{align*}
On the other hand since $\lambda\geq \eta_{\Phi, \Psi}(A, B)$, for $\E[f|A]$, as a function of $A$, we have
\begin{align*}
\lambda H_\Phi(\E[f|A]) \geq H_\Psi\big(\E\big[ \E[f|A]|B  \big]\big) = H_\Psi\big(\E[f|B  ]\big),
\end{align*}
where the equality follows from the independence of $X$ and $(A, B)$.
Summing up the above two inequalities we obtain 
\begin{align*}
\lambda H_\Phi(f) \geq H_\Psi(\E[f|AY]|A) + H_\Psi(\E[f|B  ]).
\end{align*}
Then it suffices to show that 
$$H_\Psi(\E[f|AY]|A) + H_\Psi(\E[f|B  ])\geq H_\Psi(\E[f|BY]).$$
Let $g_{AY}=\E[f|AY]$. Then we have $\E[f|B] = \E[g|B]$ and $\E[g|BY]=\E[f|BY]$. Therefore, this inequality is equivalent to
$$H_\Psi(g|A) + H_\Psi(\E[g|B])\geq H_\Psi(\E[g|BY]).$$
Using $\Psi\in \mathscr F$, this inequality follows from part~(c) of Lemma~\ref{lem:key-lemma-phi-entropy} and $H_\Psi(g|A)\geq H_\Psi(g|AB)$.

Now we assume that $\Phi\in \mathscr F$ and prove~\eqref{eq:lambda-s-phi-psi}. Since $\lambda\geq \eta_{\Phi, \Psi}(X, Y)$, and the distribution of $(X, Y)$ does not change when we condition on $B$, we have
$$ \lambda H_\Phi(\E[f|XB]|B) \geq H_\Psi\big(\E\big[\E[f|XB]|YB\big]\big|B\big)=H_\Psi(\E[f|BY]|B),$$
Similarly, since $\lambda\geq \eta_{\Phi, \Psi}(A, B)$, for $\E[f|A]$ as function of $A$ we have
$$\lambda H_\Phi(\E[f|A]) \geq H_\Psi\big(\E\big[ \E[f|A] \big|B   \big]  \big)= H_\Psi(\E[f|B]) = H_\Psi\big(\E\big[ \E[f|BY] \big|B   \big]  \big).
$$
Summing up the above two inequalities and using the chain rule we find that 
\begin{align*}
\lambda \big(  H_\Phi(\E[f|XB]|B)+    H_\Phi(\E[f|A]) \big)\geq  H_\Psi(\E[f|BY]).
\end{align*}
Then it suffices to show that 
$$H_\Phi(f)\geq H_\Phi(\E[f|XB]|B)+    H_\Phi(\E[f|A]).$$
Observe that since $\Phi\in \mathscr F$ and $A-B-X$ forms a Markov chain, by part~(b) of Lemma~\ref{lem:key-lemma-phi-entropy} we have
$H_\Phi(\E[f|XB]|B)\leq H_\Phi(f|AB)\leq H_\Phi(f|A)$. The above inequality then follows from the chain rule. 
\end{proof}

Note that in the proof of the monotonicity property we use only the convexity of $\Phi$ and $\Psi$.

%*****************************

\section{Proof of Theorem \ref{thm:newamiphi}}

\label{app:thm:newamiphi}
Given $u\in\mathcal U$, let $$f_u(x_{[k]})=\frac{p(x_{[k]},u)}{p(x_{[k]})p(u)},$$
Then, $f_u\geq 0$ and $\mathbb{E}_{p(x_{[k]})}[f_u(X_{[k]})]=1$, and 
$$\sum_{u}p(u)H_\Phi(f_u)=I_{\Phi}(X_{[k]};U)$$
Furthermore,
\begin{align}\mathbb{E}[f_u(X_{[k]})|X_i=x_i]&=\sum_{x_{\hat{i}}}p(x_{\hat{i}}|x_i)f_u(x_{[k]})
\\&=\sum_{x_{\hat{i}}}p(x_{\hat{i}}|x_i)\frac{p(x_{[k]},u)}{p(x_{[k]})p(u)}
\\&=\frac{p(x_i,u)}{p(x_i)p(u)}
\end{align}
As a result
$$\sum_{u}p(u)H_\Phi(\mathbb{E}[f_u|X_i])=I_{\Phi}(X_i;U).$$
Now, if $\lambda_{[k]}\in \fR'_{\Phi}(X_{[k]})$, for any $u$ we have
$$\sum_i\lambda_iH_{\Phi}(\mathbb{E}[f_u|X_i])\leq  H_{\Phi}(f_u).$$
Multiplying the above in $p(u)$ and adding up over $u$, we get that $\lambda_{[k]}$ is in the 
$\fR_{I_{\Phi}}(X_{[k]})$. Hence, $ \fR'_{\Phi}(X_{[k]})\subseteq \fR_{I_{\Phi}}(X_{[k]})$.

Conversely, to show that $\fR_{I_{\Phi}}(X_{[k]}) \subseteq\fR'_{\Phi}(X_{[k]})$
we adopt the construction from \cite{AGKN}. Assume that $\lambda_{[k]}$ is in $\fR_{I_{\Phi}}(X_{[k]})$. Let $f(x_1,\dots, x_k)$ be an arbitrary function satisfying $f\geq 0$ and $\mathbb{E}[f]=1$. We want to show that
$$\sum_i\lambda_iH_{\Phi}(\mathbb{E}[f|X_i])\leq  H_{\Phi}(f).$$

Without loss of generality we may assume that $f(x_1,\dots, x_n)$ is zero whenever $p(x_1\dots x_n)$ is zero. Define the distribution 
$p_{\epsilon}(u, x_1,\dots, x_k)$ as follows.
Let $U_\epsilon$ be a binary random variable such that $p_\epsilon(U_\epsilon=0)=\epsilon$ and $p_\epsilon(U_\epsilon=1)=1-\epsilon$. Also let 
\begin{align*}
p_{\epsilon}(x_1,\dots ,x_k|U_\epsilon=0)&=p(x_1,\dots ,x_k)f(x_1,\dots ,x_k)
 \\
p_{\epsilon}(x_1,\dots, x_k|U_\epsilon=1)&=p(x_1,\dots ,x_k)\left(\frac{1}{1-\epsilon}  - \frac{\epsilon}{1-\epsilon}f(x_1,\dots ,x_k)\right)
\end{align*}
Observe that for sufficiently small $\epsilon\geq 0$, $p_\epsilon(u, x_1,\dots ,x_k)$ is a probability distribution because  $f\geq 0, \mathbb{E}[f]=1$. Moreover, we have $p_{\epsilon}(x_1,\dots ,x_k)= p(x_1,\dots ,x_k)$. Then since $(\lambda_1,\dots, \lambda_k)$ is in $I_\Phi$-ribbon, we have 
$$\sum_i \lambda_i I_\Phi(U_\epsilon; X_i)\leq I_\Phi(U_\epsilon; X_1\dots X_k).$$
Indeed the function 
$$t(\epsilon)= I_\Phi(U_\epsilon; X_1\dots X_k) - \sum_i \lambda_i I_\Phi(U_\epsilon; X_i)$$
is non-negative for sufficiently small $|\epsilon|$. On the other hand, we have $t(0)=0$. Then we should have $t'(0)\geq 0$. Observe that
\begin{align}I_\Phi(U_\epsilon; X_1\dots X_k)&=\sum_{x_{[k]}}\epsilon p(x_{[k]})\Phi(\frac{\epsilon p(x_{[k]})f(x_{[k]})  }{\epsilon p(x_{[k]})})+\sum_{x_{[k]}}(1-\epsilon) p(x_{[k]})\Phi(\frac{
p(x_{[k]})\left(1- \epsilon f(x_{[k]})\right)
 }{(1-\epsilon) p(x_{[k]})})
\\&=\sum_{x_{[k]}}\epsilon p(x_{[k]})\Phi(f(x_{[k]}))+\sum_{x_{[k]}}(1-\epsilon) p(x_{[k]})\Phi(\frac{
1- \epsilon f(x_{[k]})
 }{1-\epsilon})
\end{align}
Then,\begin{align}\frac{\partial}{\partial \epsilon}I_\Phi(U_\epsilon; X_1\dots X_k)\Big|_{\epsilon=0} &=\sum_{x_{[k]}}p(x_{[k]})\Phi(f(x_{[k]}))-\Phi(1)
\\&=H_{\Phi}(f)
\end{align}
Now, observe that
\begin{align*}
p_{\epsilon}(x_i|U_\epsilon=0)&=p(x_i)\mathbb{E}[f|X_i=x_i]
 \\
p_{\epsilon}(x_i|U_\epsilon=1)&=p(x_i)\left(\frac{1}{1-\epsilon}  - \frac{\epsilon}{1-\epsilon}\mathbb{E}[f|X_i=x_i]\right)
\end{align*}
Thus, similarly,
$$\frac{\partial}{\partial \epsilon}I_\Phi(U_\epsilon; X_i)\Big|_{\epsilon=0} = H_{\Phi}(\mathbb{E}[f|X_i]).$$
Putting these together we obtain the desired equation.

%*****************************
\section{Proof of Theorem \ref{thm:HC-to-MC}}\label{app:HC-to-MC}
The equality $ \fR_{\Phi_\alpha}(X_{[k]})=\fR_{\varphi_\alpha}(X_{[k]})$ for $\alpha=2$ is clear since $\Phi_2=\varphi_2$. So we assume that $\alpha\in (1, 2)$.
Moreover, the inclusion  $ \fR_{\Phi_\alpha}(X_{[k]})\subseteq \fR_{\varphi_\alpha}(X_{[k]})$ is immediate once we note that 
$$\Phi_\alpha(t) = c_1\big(\varphi_\alpha(1+t) +\varphi_\alpha(1-t)\big) - c_2,$$ 
for some positive constants $c_1, c_2$. Then for any function $f_{X_{[k]}}$, we have
$$H_{\Phi_\alpha}(f) = c_1\big(H_{\varphi_\alpha}(1+f)+H_{\varphi_\alpha}(1-f)\big).$$
Writing the above equation for $f$, and $\mathbb{E}[f|X_i]$ for $i\in[k]$, we obtain $ \fR_{\varphi_\alpha}(X_{[k]})\subseteq \fR_{\Phi_\alpha}(X_{[k]})$. 
So we need to prove the inclusion in the other direction.

Let $\lambda_{[k]}\in \fR_{\Phi_\alpha}(X_{[k]})$. We will show that $\lambda_{[k]}\in \fR_{\varphi_\alpha}(X_{[k]})$. 
Let $f\geq 0$ be some non-negative function with $m=\E f$. Define  
$$g_\epsilon = \epsilon f - 1.$$
Then for sufficiently small $\epsilon\geq 0$ the range of  $g_\epsilon$ is inside the domain of $\Phi_\alpha$. For any such $\epsilon$ we have 
\begin{align}\label{eq:H-Phi-alpha-varphi}
H_{\Phi_\alpha}(g_\epsilon) \geq \sum_{i=1}^k\lambda_i H_{\Phi_\alpha}(\E[g_\epsilon|X_i]).
\end{align}
We compute 
\begin{align*}
\Phi_\alpha(g_\epsilon) - \Phi_\alpha(\E g_\epsilon) & = \frac{1}{2^\alpha-2} \Big(  (1+g_\epsilon)^\alpha - (1+\E g_\epsilon)^\alpha + (1-g_\epsilon)^\alpha - (1-\E g_\epsilon)^\alpha     \Big)\\
& =  \frac{1}{2^\alpha-2} \Big(  (\epsilon f)^\alpha - (\epsilon m )^\alpha + (2-\epsilon f)^\alpha - (2-\epsilon m )^\alpha     \Big)\\
& = \frac{1}{2^\alpha-2} \Big(  \epsilon^\alpha ( f^\alpha - m ^\alpha) + (2-\epsilon f)^\alpha - (2-\epsilon m )^\alpha     \Big).
\end{align*}
Taking expectation from both sides we obtain
\begin{align*}
H_{\Phi_\alpha}(g_\epsilon) & = \frac{1}{2^\alpha-2} \Big( \epsilon^\alpha (\E[f^\alpha] - m ^\alpha) + \E[(2-\epsilon f)^\alpha - (2-\epsilon m )^\alpha   ] \Big)\\ 
& = \frac{\epsilon^\alpha}{2^\alpha-2}  H_{\varphi_\alpha}(f) + O(\epsilon^2),
\end{align*}
where the $O(\epsilon^2)$ term is  derived once we take the Taylor expansion of $(2-t)^\alpha$.
We similarly have
\begin{align*}
H_{\Phi_\alpha}\big(\E[g_\epsilon|X_i]\big)& = \frac{\epsilon^\alpha}{2^\alpha-2} 
H_{\varphi_\alpha}\big(\E[f|X_i]\big) + O(\epsilon^2).
\end{align*}
Putting these in~\eqref{eq:H-Phi-alpha-varphi} and using the fact that $1<\alpha<2$ we find that 
$$H_{\varphi_\alpha}(f) \geq \sum_{i=1}^k\lambda_i H_{\varphi_\alpha}(\E[f|X_i]).$$
Therefore, $\lambda_{[k]}\in \fR_{\varphi_\alpha}(X_{[k]})$.

%*********************************************************************************************
%%%%%%%%%%%%%%%
\iffalse
\section{Proof of Theorem \ref{thm:Phi-ribbon-MC-ribbon}}\label{appendix:thm-Phi-ribbon-MC-ribbon}
Let $\lambda_{[k]}\in \fR_{\Phi}(X_1, \dots, X_k)$. We will show that $\lambda_{[k]}\in \fS(X_1, \dots, X_k)$. For this we need to show that for any $f_{X_{[k]}}$ with $\E f=0$ we have
\begin{align}\label{eq:goal-proof-mc-phi}
\Var[f] \geq \sum_{i=1}^k\lambda_i\Var_{X_i}\big[\E[f|X_i]\big].
\end{align} 

Let $c\in \mathbb R$ be in the interior of the domain of $\Phi$. Then for any $\epsilon\in \mathbb R$ with sufficiently small $|\epsilon|> 0$, the range of the function $g_{X_{[k]}} = c+ \epsilon f_{X_{[k]}}$  is in the interior of the domain of $\Phi$. Then for any such $\epsilon$ we have 
\begin{align}\label{eq:g-c-var-phi}
H_\Phi(g) \geq \sum_{i=1}^k\lambda_iH_\Phi\big(\E[g|X_i]\big).
\end{align}
Now writing down the Taylor expansion of $\Phi$ around $c$, we find that 
$$H_\Phi(g) = \frac{1}{2} \Phi''(c)\Var[f]\epsilon^2 + O(|\epsilon|^3),$$
and
$$H_\Phi\big(\E [g|X_i]\big) = \frac{1}{2} \Phi''(c)\Var[f|X_i]\epsilon^2 + O(|\epsilon|^3),$$
Putting these in~\eqref{eq:g-c-var-phi} and taking the limit $\epsilon\rightarrow 0$ gives~\eqref{eq:goal-proof-mc-phi}.

\fi
%%%%%%%%%%%%
%*********************************************************
\section{Proof of Theorem \ref{thm:2-rep-MC-ribbon}}\label{app:proof-thm:2-rep-MC-ribbon}

Using Proposition~\ref{prop:lemmaMCrb} and  the inequality $\Var[\E[f|X_i]]\geq \E[f\hat f_i]^2$, the inclusion $\fS(X_{[k]})\subseteq \fS'(X_{[k]})$ is immediate. It suffices to show that $\fS'(X_{[k]})\subseteq \fS(X_{[k]})$.

Let $\lambda_{[k]}\in \fS'(X_{[k]})$ and let $f_i(X_i)$, $i=1,\dots, k$, be arbitrary functions with zero mean. According to the characterization of Theorem~\ref{thm:mc-equivalent-rep} of the MC~ribbon we need to show that 
$$\Var[f_1+\cdots +f_k]\leq  \sum_{i=1}^k \frac{1}{\lambda_i} \Var[f_i].$$

Let  $c_i=\sqrt{\Var[f_i]}$ and define $\hat f_i= f/c_i$. Also define
 $m_{ij}=\E[\hat f_i\hat f_j]$. Observe that
$$\Var[f]=\sum_{i,j=1}^kc_{i}c_{j}m_{ij},$$
and 
\begin{align*}
\E[f\hat f_i]^2=\bigg(\sum_{j=1}^kc_jm_{ij}\bigg)^2.
\end{align*}
Now since  $\lambda_{[k]}\in \fS'(X_{[k]})$ we have 
$$\sum_{i,j=1}^kc_{i}c_{j}m_{ij}\geq \sum_{i=1}^k \lambda_i \bigg(\sum_{j=1}^kc_jm_{ij}\bigg)^2.$$
Indeed, this inequality must hold for all choices of $c_i$'s.

Let $M$ be a matrix whose $(i,j)$-th entry is $m_{ij}$. Also let $\Lambda$ be a diagonal matrix with diagonal entries equal to $\lambda_1, \dots, \lambda_k$.
Then, the above inequality, for all choices of $c_i$'s, is equivalent to
$M\geq M\Lambda M$, which itself is equivalent to $M^{-1}\geq \Lambda$. Next, since $t\mapsto -t^{-1}$ is operator monotone, it is also equivalent to $\Lambda^{-1}\geq M$. Then by simple calculation $\Lambda^{-1}\geq M$ means that for all $c_1, \dots, c_k$ we have 
$$\Var[c_1\hat f_1+\cdots + c_k\hat f_k]\leq \sum_{i=1}^k \frac{1}{\lambda_i} c_i^2 =  \sum_{i=1}^k \frac{1}{\lambda_i} \Var[c_i \hat f_i].$$
This is what we wanted to show. 

%************************************************
\section{Proof of Proposition~\ref{prop:binary-binary-ternary}}\label{app:binary-binary-ternary}

We use Theorem~\ref{thm:mc-equivalent-rep} to prove this proposition. 
Any zero-mean function of $X_i$, $i=1, 2$, is of the form $f_i=a_ig_i$ for some constants $a_1, a_2$. Moreover, the space of zero-mean functions of $X_3$ is two-dimensional. 
Since we assume that $\mathbb{E}[g_1|X_3]$ and $\mathbb{E}[g_2|X_3]$ are linearly independent, this space is spanned by these two functions. That is,  any zero-mean function $f_3(X_3)$ can be expressed as 
$$f_3=a_3  \mathbb{E}[g_1|X_3]+a_4  \mathbb{E}[g_2|X_3],$$ 
for some constants $a_3$ and  $a_4$. 
Then using equations~\eqref{eq:g1g2-E}-\eqref{eq:def-r123}  we have
\begin{align*}
\Var[f_1]&=a_1^2,
\\ \Var[f_2]&=a_2^2,
\\\Var[f]&=a_3 ^2\rho_{13}^2+a_4 ^2\rho_{23}^2+2a_3 a_4  r_{12\rightarrow 3},
\\ \E[f_1 f_2]&=a_1a_2\rho_{12},
\\ \E[f_1f_3]&=a_1\big(a_3  \rho_{13}^2+a_4  r_{1,2\rightarrow 3}\big),
\\ \E[f_2f_3]&=a_2\big(a_3  r_{1,2\rightarrow 3}+a_4  \rho_{2,3}^2\big).
\end{align*}
Hence, 
\begin{align*}
\Var\bigg[\sum_{i=1}^k f_i\bigg]&=a_1^2+a_2^2+a_3 ^2\rho_{1,3}^2+a_4 ^2\rho_{2,3}^2+2a_3 a_4  r_{1,2\rightarrow 3}
\\&\quad\, +2a_1a_2\rho_{1,2}+2a_1\big(a_3  \rho_{1,3}^2+a_4  r_{1,2\rightarrow 3}\big)+2a_2\big(a_3  r_{1,2\rightarrow 3}+a_4  \rho_{2,3}^2\big).
\end{align*}
On the other hand, $\fS(X_1, X_2, X_3)$ is the set of triples $(\lambda_1,\lambda_2, \lambda_3)$ such that
$$\Var \bigg[\sum_{i=1}^k f_i\bigg]\leq \sum_{i=1}^3 \frac{1}{\lambda_i} \Var[f_i],$$
for all choices of $a_1, \dots, a_4$.
Using the previous equations, the above inequality is a quadratic form is terms of $a_1, \dots, a_4$. Indeed, letting $\mathbf{v}=[a_1, a_2, a_3 , a_4 ]$, the above inequality is equivalent to $\mathbf{v} \Delta \mathbf{v}^t\geq 0$ for all $\mathbf{v}$, where
\begin{align}\label{eqn:AAA22224}
\Delta=\begin{bmatrix}
   \frac{1}{\lambda_1}-1       &-\rho_{12} & -\rho_{13}^2 & -r_{12\rightarrow 3} \\
    -\rho_{12}       & \frac{1}{\lambda_2}-1 & -r_{12\rightarrow 3} & -\rho_{23}^2\\
    -\rho_{13}^2       &-r_{1,2\rightarrow 3} & \rho_{13}^2(\frac{1}{\lambda_3}-1)  &  r_{1,2\rightarrow 3}(\frac{1}{\lambda_3}-1) \\
    -r_{1,2\rightarrow 3}       & -\rho_{23}^2 &  r_{1,2\rightarrow 3}(\frac{1}{\lambda_3}-1) & \rho_{23}^2(\frac{1}{\lambda_3}-1)
\end{bmatrix}.
\end{align}
 In other words, $\fS(X_1, X_2, X_3)$ is the set of $(\lambda_1,\lambda_2, \lambda_3)$ for which $\Delta$ is positive semi-definite.

Observe that $\Delta$ can be written in the block form:
\[
\Delta=\begin{bmatrix}
    A & - B
\\- B& (\frac{1}{\lambda_3}-1) B
\end{bmatrix},
\]
where $A$ and $B$ are $2\times 2$ matrices.
Then using \cite[p.14]{Bhatia},  $\Delta$ is positive semi-definite if and only if 
$$ A- B(\frac{1}{\lambda_3}-1)^{-1}=
\begin{bmatrix}
   \frac{1}{\lambda_1}-1-\rho_{1,3}^2(\frac{1}{\lambda_3}-1) ^{-1}        &-\rho_{1,2} -r_{1,2\rightarrow 3}(\frac{1}{\lambda_3}-1)^{-1}  \\
    -\rho_{1,2}  -r_{1,2\rightarrow 3}(\frac{1}{\lambda_3}-1)^{-1}     & \frac{1}{\lambda_2}-1 -\rho_{2,3}^2(\frac{1}{\lambda_3}-1)^{-1} 
\end{bmatrix},$$ 
is positive semi-definite. This is equivalent with the conditions given in the statement of the theorem.
%********************************************************************
\section{Proof of Theorem \ref{thm:Gaussian}}\label{app:Gaussian}

When $X_i$'s are binary, then any function of $X_i$ with zero mean is of the form $f_i(X_i)= a_i(X_i-\E[X_i])$ for some constant $a_i$. Using this fact it is easy to see that 
\begin{align}\label{eq:var-f-i-01}
\Var\bigg[\sum_i f_i\bigg]\leq \sum_i \frac{1}{\lambda_i}\Var[f_i],
\end{align}
holds for all choices of $a_i$'s if and only if $R\leq \Lambda^{-1}$.

Let us turn to the proof for Gaussian variables. Our proof is an extension of the proof of Lancaster~\cite{Lancaster57} to the multivariate case.

Observe that scaling of and adding a constant to the variables $X_i$ would not change the MC~ribbon. Hence, without loss of generality we assume that $\E[X_i]=0$, $\Var[X_i]=1$, and $\mathbb{E}[X_iX_j]=R_{ij}$.

The Hermite-Tchebycheff polynomials are  defined as follows:
$$\psi_\ell(x)=(-1)^\ell e^{x^2}\frac{d^\ell}{dx^\ell}e^{-x^2}, \qquad \ell\geq 0.$$
The following facts are known about  these polynomials~\cite{Lancaster57}:
\begin{enumerate}[label=(\roman*)]
\item $\psi_0(x)=1$ is a constant function, and $\psi_i(x)$ and $\psi_j(x)$ are  orthonormal with respect to standard normal distribution, i.e., 
$$\frac{1}{2\pi}\int_{-\infty}^{\infty}\psi_i(x)\psi_j(x)e^{-\frac{x^2}{2}}\dd x=\delta_{ij}.\qquad \forall i,j.$$

\item If $X$ is a normal random variable, any function of $X$ denoted by $f(X)$ that has finite variance can be approximated as follows: for any $\epsilon>0$, there is a sequence $\{a_\ell | \,\ell\geq 0\}$ such that $\sum_{\ell} a_\ell^2$ is convergent, and for
$$\hat{f}(x)=\sum_{\ell=0}^\infty a_\ell\psi_\ell(x),$$
we have 
$$\mathbb{E}\Big[\big|f(X)-\hat{f}(X)\big|^2\Big]\leq \epsilon.$$
Furthermore, if $\mathbb{E}[f(X)]=0$, we may take $a_0=0$.

\item If $X$ and $Y$ are unit variance, jointly Gaussian random variables with correlation coefficient $\rho$, then 
$$\mathbb{E}[\psi_\ell(X)\psi_{\ell'}(Y)]=\delta_{\ell \ell'}\rho^\ell.$$
\end{enumerate}

Fix $(\lambda_1, \dots, \lambda_k)\in [0,1]^k$. To verify the validity of~\eqref{eq:var-f-i-01} take
some arbitrary zero-mean functions $f_i(X_i)$, $i=1, \dots, k$, with finite variance. Fix some $\epsilon>0$, and using property~(ii) explained above let
\begin{align}
\hat{f}_i(x_i)=\sum_{\ell=1}^\infty a_{i\ell}\psi_\ell(x_i),
\label{eqnappx}
\end{align}
be such that
\begin{align}\mathbb{E}\Big[\big|f_i(X_i)-\hat{f}_i(X_i)\big|^2\Big]\leq \epsilon, \qquad \forall i.
\label{eqn:edited1}\end{align}
Then, by the Cauchy-Schwarz inequality we have
\begin{align}\mathbb{E}\bigg[\Big| \sum_{i=1}^k f_i(X_i)- \sum_{i=1}^k\hat{f}_i(X_i)\Big|^2\bigg]\leq k\epsilon.\label{eqn:edited2}\end{align}
Then, it is not hard to verify that \eqref{eqn:edited1} implies
$$\Big| \Var[\hat{f_i}(X_i)]-\Var[f_i(X_i)]\Big|=O(\sqrt \epsilon),$$
and similarly \eqref{eqn:edited2} implies
$$\bigg| \Var\Big[ \sum_{i=1}^k \hat f_i(X_i)\Big] -\Var\Big[\sum_{i=1}^k {f}_i(X_i)\Big]\bigg|=O(\sqrt \epsilon).$$
Therefore, it suffices to verify~\eqref{eq:var-f-i-01} for functions of the form~\eqref{eqnappx}. 

Using properties (i) and (iii) we have 
$$\Var[\hat{f}_i(X_i)]=\sum_{\ell=1}^\infty a_{i\ell}^2,$$
and
\begin{align*}
\Var\Big[\sum_{i=1}^k\hat{f}_i(X_i)\Big]&=\sum_{i_1,i_2=1}^k \E\big[\hat{f}_{i_1}(X_{i_1})\hat{f}_{i_2}(X_{i_2})\big]
\\&=\sum_{i_1,i_2=1}^k \sum_{\ell_1, \ell_2=1}^\infty a_{i_1\ell_1}a_{i_2\ell_2}\E\big[\psi_{\ell_1}(X_{i_1})\psi_{\ell_2}(X_{i_2})\big]
\\&=\sum_{i_1,i_2=1}^k\sum_{\ell=1}^\infty a_{i_1\ell}a_{i_2\ell}R_{i_1, i_2}^\ell.
\end{align*}
Thus, we are interested in the set of $k$-tuples $(\lambda_1, \dots, \lambda_k)$ such that for all $a_{i\ell}$'s we have 
$$\sum_{i_1,i_2=1}^k\sum_{\ell=1}^\infty a_{i_1\ell}a_{i_2\ell}R_{i_1, i_2}^\ell \leq \sum_{i=1}^k \frac{1}{\lambda_i}\sum_{\ell=1}^\infty a_{i\ell}^2.$$
This holds if and only if for any $\ell\in\mathbb{N}$ and for any $a_{i\ell}$'s  we have
$$\sum_{i_1,i_2=1}^k a_{i_1\ell}a_{i_2\ell}R_{i_1, i_2}^\ell \leq \sum_{i=1}^k \frac{1}{\lambda_i}a_{i\ell}^2.$$
This can be expressed in matrix form as 
\begin{align}
R^{\circ \ell}\leq \Lambda^{-1} \qquad \forall \ell\geq 1,\label{eqn:Had}\end{align}
where $R^{\circ \ell}$ is the Hadamard product (entry-wise product) of $R$ with itself $\ell$ times. For $\ell=1$, we have the condition
\begin{align}
R\leq \Lambda^{-1}.
\label{eqn:Had2}
\end{align}
Now, we claim that \eqref{eqn:Had2} implies \eqref{eqn:Had} for any $\ell\geq 2$. To prove this, note that $R\leq \Lambda^{-1}$ means that $\Lambda^{-1}-R\geq 0$ is positive semi-definite. Moreover, $R$ is a correlation matrix, so it is positive semi-definite. Since the Hadamard product of two positive semi-definite matrix is positive semi-definite, $R^{\circ (\ell -1)}$ is positive semi-definite as well. Similarly, $R^{\circ (\ell -1)}\circ(\Lambda^{-1}-R)$ is positive semi-definite, \emph{i.e.}, $R^{\circ (\ell -1)} \circ \Lambda^{-1}\geq R^{\circ \ell}$. Now the point is that the diagonal entries of $R$ are all one, and $\Lambda$ is diagonal. Therefore, $R^{\circ (\ell -1)} \circ \Lambda^{-1} = \Lambda^{-1}$. This completes the proof.

%******************************************************

\section{Proofs of Theorem \ref{thm:tensorization2}}\label{app:tilde-S-proof}

\paragraph{Tensorization:}
Assuming that $X_{[k]}$ and $Y_{[k]}$ are independent, we would like to show that 
$$\widetilde{\fS}(X_1Y_1,\cdots, X_kY_k)=\widetilde{\fS}(X_1,\cdots, X_k)\cap\fS(Y_1,\cdots, Y_k)$$
We clearly have
$$\widetilde{\fS}(X_1Y_1,\cdots, X_kY_k)\subseteq \widetilde{\fS}(X_1,\cdots, X_k)\cap\fS(Y_1,\cdots, Y_k)$$
since we can restrict to functions $f_i(X_{i},Y_{i})$ to depend only on one of $X_{i}, Y_{i}$. To prove the inclusion in the other direction, take some $(\lambda_1, \cdots, \lambda_k)\in \widetilde{\fS}(X_1,\cdots, X_k)\cap\widetilde{\fS}(Y_1,\cdots, Y_k)$. For arbitrary functions $f_i(X_i, Y_i)$, $i=1,\dots, k$, we compute
\begin{align}
\Var_{X_{[k]}Y_{[k]}}\bigg[\sum_{i=1}^k f_i\bigg]&=
\Var\Bigg[\E\Big[\sum_{i=1}^kf_i\Big|X_{[k]}\Big]\Bigg]+
\Var\Big[\sum_{i=1}^kf_i\Big| X_{[k]}\Big]\label{eqn:totalvar1}
\\&=\Var\Bigg[\sum_{i=1}^k\E\Big[f_i\Big|X_{i}\Big]\Bigg]+
\Var\Big[\sum_{i=1}^kf_i\Big| X_{[k]}\Big]\label{eqn:totalvar2}
\\&\geq \sum_{i=1}^k{\lambda_i}\Var\big[\E[f_i|X_i]+
\sum_{i=1}^k{\lambda_i}\Var[f_i|X_{[k]}]\label{eqn:totalvar3}\\
&=\sum_{i=1}^k{\lambda_i}\Var\big[\E[f_i|X_i]+
\sum_{i=1}^k{\lambda_i}\Var[f_i|X_{i}]\label{eqn:totalvar5}
\\&=\sum_{i=1}^k{\lambda_i}\Var[f_i].\label{eqn:totalvar4}
\end{align}
Here equations~\eqref{eqn:totalvar1} and \eqref{eqn:totalvar4} follow from the law of total variance. Equations~\eqref{eqn:totalvar2} and~\eqref{eqn:totalvar5}  follow from the independence of $X_{[k]}$ and $Y_{[k]}$. Finally,~\eqref{eqn:totalvar3} holds since $(\lambda_1, \dots, \lambda_k)$ is in both $\widetilde \fS(X_{[k]})$
and $\widetilde \fS(Y_{[k]})$.

\paragraph{Monotonicity:}
Let $(\lambda_1, \dots, \lambda_k)\in \widetilde\fS(X_1, \dots, X_k)$ and let $f_i(Y_i)$, $i=1, \dots, k$, be arbitrary functions. We need to show that 
\begin{align}\label{eq:mon-29ij}
\Var\Big[ \sum_{i=1}^k f_i \Big] \geq \sum_{i=1}^k \lambda_i \Var[f_i].
\end{align}

For functions $\E[f_i|X_i]$, $i=1, \dots, k$, we have 
$$\Var\Bigg[\E\Big[\sum_{i=1}^k  f_i\Big|X_{[k]}\Big] \Bigg] =\Var\Big[\sum_{i=1}^k  \E[f_i|X_i] \Big] \geq \sum_{i=1}^k \lambda_i \Var[\E[f_i|X_i]].$$
Moreover, since $Y_i$'s are independent conditioned on $X_{[k]}$ we have 
$$\Var\Big[   \sum_{i=1}^k f_i \Big| X_{[k]}     \Big]=  \sum_{i=1}^k \Var[f_i|X_{[k]}] = \sum_{i=1}^k \Var[f_i|X_i] \geq \sum_{i=1}^k\lambda_i \Var[f_i|X_i].$$
Summing up the above two inequalities and using the law of total variance, we obtain~\eqref{eq:mon-29ij}. 

%************************************************************
\section{Proof of Theorem \ref{thm:tilde-S-extreme-0}}\label{app:tilde-S-extreme-0}

If $(f_1, \dots, f_k)$ with the conditions in the theorem exists, we clearly have $\widetilde \fS(X_{[k]})=\{(0,\dots, 0)\}$.

Now suppose that $\widetilde \fS(X_{[k]})=\{(0,\dots, 0)\}$. Then for any $\epsilon>0$ there are functions $f_i^{(\epsilon)}(X_i)$, $i=1,\dots, k$, such that 
$$\Var\Big[f_1^{(\epsilon)}+ \cdots+ f_k^{(\epsilon)}\Big] \leq \epsilon \sum_{i=1}^k \Var\Big[f_i^{(\epsilon)}\Big],$$
and 
$$\sum_{i=1}^k \Var\Big[f_i^{(\epsilon)}\Big]=1.$$
Then by a compactness argument, there are limiting functions $\hat f_i(X_i)$, $i=1,\dots, k$, such that 
$$\sum_{i=1}^k \Var\big[\hat f_i\big]=1,$$
and 
$$\Var[\hat f_1+\cdots +\hat f_k]=0.$$
Since $\E[\hat f_i]=0$ for all $i$, the latter equation means that $\hat f_1+\cdots + \hat f_k=0$. Furthermore, $(\hat f_1, \dots, \hat f_k)$ is non-zero because of $\sum_{i=1}^k \Var[\hat f_i]=1.$
We are done.

\end{document}